\documentclass[pra,twocolumn,superscriptaddress,floatfix,amsmath,amssymb,longbibliography]{revtex4-2}
\pdfoutput=1

\usepackage[colorlinks=true]{hyperref}
\usepackage{dcolumn}
\usepackage{bm}
\usepackage{amsmath}
\usepackage{amssymb}
\usepackage{amsthm}
\usepackage{mathtools}
\usepackage{braket} 
\usepackage[dvipdfmx]{color}
\usepackage{graphicx}
\usepackage{longtable}

\usepackage{algorithm}
\usepackage{algpseudocode}

\usepackage[toc,page]{appendix}
\usepackage{makecell}

\def\Q{{\mathbb Q}}

\newcommand{\defeq}{\vcentcolon=}

\newcommand\restr[2]{{
 \left.\kern-\nulldelimiterspace 
 #1 
 \vphantom{\big|} 
 \right|_{#2} 
 }}

\newcommand{\wh}{\widehat}

\newtheorem{theorem}{Theorem}

\begin{document}

\title{A Tailored Fidelity Estimation and Purification Method for Entangled Quantum Networks
\thanks{
This work was supported by JST SPRING, Japan Grant Number JPMJSP2123 and by JST Moonshot R$\&$D Grants JPMJMS2061 and JPMJMS226C. 
}
}
\author{Takafumi Oka}
\author{Michal Hajdu\v{s}ek}
\affiliation{%
Graduate School of Media and Governance, Keio University Shonan Fujisawa Campus, Kanagawa, Japan
}%
\author{Shota Nagayama}
\affiliation{
Graduate School of Media Design, Keio University Hiyoshi Campus, Kanagawa, Japan
}%
\author{Rodney Van Meter}
\affiliation{
Faculty of Environment and Information Studies, Keio University Shonan Fujisawa Campus, Kanagawa, Japan
}%

\begin{abstract}
We present a method to conduct both quantum state reconstruction and entanglement purification simultaneously that is advantageous in several respects over previous work in this direction, showing that the number of Bell pairs necessary to boot a quantum network can be significantly reduced compared to an existing method. The existing method requires at least $10^5$ Bell pairs for the state reconstruction phase to estimate that the state is of fidelity $0.99$ within the error range of $10^{-2}$, whereas our approach only requires around $2,841$ to be certain with $99.7\%$ of confidence that the estimated fidelity lies within $[0.99-0.01, 0.99+0.01]$. In addition, in our approach we can start with a lower fidelity Bell pair and purify it multiple times, estimating at the same time the resultant fidelity with guarantee of $99.7\%$ that the fidelity estimate lies within a certain range. Moreover, the existing method cannot correct both bit-flip and phase-flip errors at the same time and can only correct one of these, whereas our approach can correct both bit-flip and phase-flip errors simultaneously. This research produces numerical estimates for the number of Bell pairs actually needed to guarantee a certain threshold fidelity $F$. The research can support the functioning real-world quantum networking by providing the information of the time needed for the bootstrapping of a quantum network to finish.

\end{abstract}
\maketitle


\section{\label{sec:introduction}Introduction}
Quantum networking and quantum computing depend on extensive use of \textit{Bell pairs}, where a Bell pair is represented as one of the four two-qubit quantum states $\ket{\Phi^+}=\frac{1}{\sqrt{2}}(\ket{00}+\ket{11})$, $\ket{\Psi^+}=\frac{1}{\sqrt{2}}(\ket{01}+\ket{10})$, $\ket {\Psi^-}=\frac{1}{\sqrt{2}}(\ket{01}-\ket{10})$, and $\ket{\Phi^-}=\frac{1}{\sqrt{2}}(\ket{00}-\ket{11})$. In the real-world setting, however, there are errors on the quantum state, and a given Bell pair is imperfect. The state is in a \textit{mixed state}, as opposed to a \textit{pure state}. One quantity that measures the degree to which it is close to the perfect state is \textit{fidelity}, which takes on real values between $0$ and $1$, with $1$ signifying the perfect state, where the given state $\rho$ can be described informally as $\rho=F\ket{\Phi^+}\bra{\Phi^+}+(1-F)(noise)$.

In operational networks, we will have to achieve a threshold fidelity suitable for distributed computation in both data center and wide area interconnects \cite{BARRAL2025100747, PRXQuantum.2.017002}. Thus, we need an online protocol that is accurate enough to make operational decisions on the disposition of a Bell pair at hand that results in one of the following actions: a) purify \cite{PhysRevLett.76.722, PhysRevA.54.3824, PhysRevA.80.042308, 10.1109/JSAC.2024.3380094, Kim_2025, Pan2001, PhysRevA.64.012304} more; b) use for entanglement swapping \cite{PhysRevLett.71.4287} or deliver to application; or c) discard. 

One means to determine the fidelity of a given Bell pair is \textit{state tomography} \cite{ALTEPETER2005105}, which attempts to fully reconstruct a given quantum state. Quantum state tomography itself has been realized and conducted in a lab in a stand-alone fashion~\cite{ALTEPETER2005105, PhysRevA.55.R1561, PhysRevA.64.052312, Qi2017, Guţă_2020}. 
Our prior work \cite{Oka2016} indicates that it is practically infeasible to employ state tomography in the real-time bootstrapping of a quantum network because it needs an excessive number of Bell pairs in order to achieve reasonable accuracy, thereby consuming too much time. Instead, it is necessary to devise and employ a more efficient fidelity estimation method specifically designed for bootstrapping a quantum network. In the first author's bachelor's thesis~\cite{OkaBThesis}, such a method was proposed in the form of a quantum circuit.

\begin{figure}
\includegraphics[keepaspectratio,width=\columnwidth]{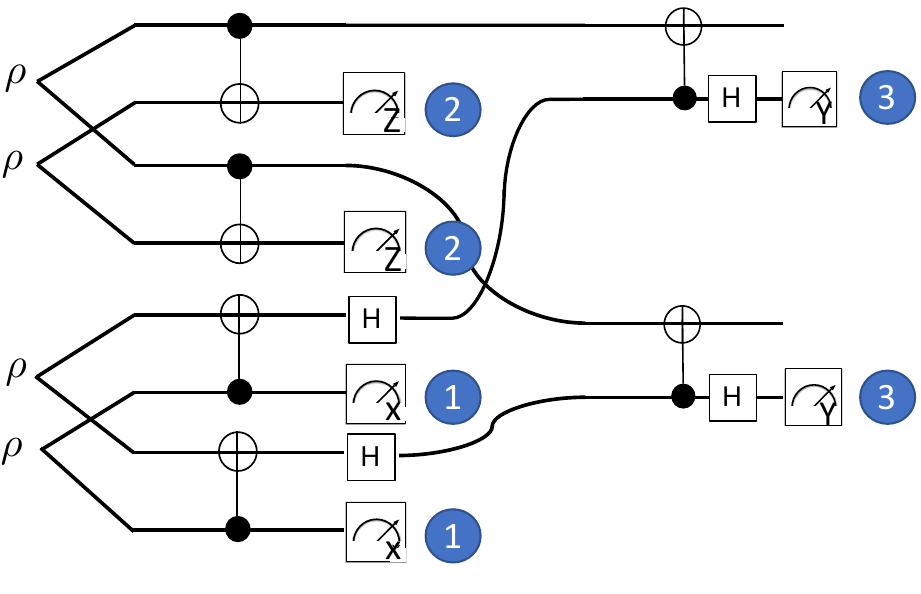}
\caption{\label{fig:the_circuit_with_number}The quantum circuit that realizes both purification and fidelity estimation concurrently. The circled numbers are measurement labels as discussed in the main text.}
\end{figure}

\begin{figure}
\includegraphics[keepaspectratio,scale=0.50]{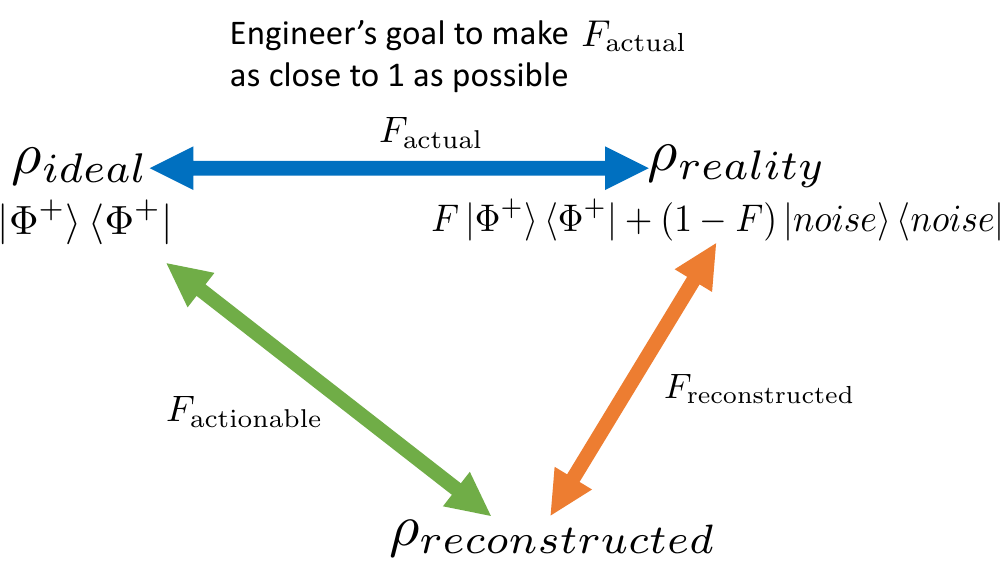}
\caption{\label{fig:OkaTriangle}Structure of the tomography process. Here $\ket{\Phi^+}=\frac{\ket{00}+\ket{11}}{\sqrt{2}}$. It is our goal to make $F_\textrm{reconstructed}$ as close to $1$ as possible. Since we make decisions based on $F_\textrm{actionable}$, when $F_\textrm{reconstructed}$ is low we make bad decisions, affecting the probability of success of application tasks.}
\end{figure}

Note that, in an experimental or operational setup, we only have access to the density matrix as reconstructed from a finite set of measurements. It is crucial to bound the distance between that density matrix, which we refer to as $\rho_{\textrm{reconstructed}}$, and the reality of the apparatus, which we refer to as $\rho_{\textrm{reality}}$, as depicted in Figure~\ref{fig:OkaTriangle}. Since we make decisions based on $F_\textrm{actionable}$, when $F_\textrm{reconstructed}$ is low we make bad decisions. Of particular interest to operational work, the technique of Sugiyama \textit{et al.} attempts to conduct tomography within a prescribed error bound \cite{PhysRevLett.111.160406}.
A related method of fidelity estimation (in contrast to full state tomography) was proposed in \cite{PhysRevLett.106.230501, PhysRevLett.107.210404}. While fidelity estimation yields much less information than state tomography, ``it saves tremendously in measurement and sample complexity'' \cite{Eisert:2020aa}.

Our solution, distinct from stand-alone state tomography, is to incorporate characterization of the state into the process of purification and conduct it in a real-time, distributed fashion. State reconstruction and in particular fidelity estimation becomes part of the process of purification, although with the caveat that the target states are restricted to Bell-diagonal states.  They function collaboratively to produce a high-fidelity Bell pair and guarantee that it is indeed of high fidelity, for a certain threshold fidelity $F$, e.g. allowing us to designate a goal of $F=0.9$ or $F=0.99$.

Prior research in the direction of conducting both purification and fidelity estimation concurrently includes Casapao \textit{et al.}'s work \cite{Casapao:2025aa} and Maity \textit{et al.}'s work \cite{10.1145/3626570.3626594}. However, Maity \textit{et al.}'s work assumes the input states are Werner states, whereas our approach does not impose that assumption. Furthermore, using Casapao \textit{et al.}'s approach one can only correct either of bit-flip and phase-flip errors in the scheme, leaving the other error in the state. Moreover, with Casapao \textit{et al.}'s approach we require at least $10^5$ Bell pairs for the state reconstruction phase to estimate that the state is of fidelity $0.99$ within $99\%$ of confidence. In contrast, in this work we have employed a distinct circuit, which enables correction of both bit-flip and phase-flip errors at the same time as well as estimation of the given state. Moreover, our approach only requires around $2,800$ Bell pairs to be $99.7\%$ certain that the fidelity estimate of the true Werner state of fidelity $F=0.99$ lies within $[0.99-0.01, 0.99+0.01]$, and $16,000$ Bell pairs to be $99.7\%$ certain that the fidelity estimate of the true Werner state of fidelity $F=0.95$ lies within $[0.95-0.01, 0.95+0.01]$. Our approach generally outperforms Casapao \textit{et al.}'s method in the high-fidelity region ($F \geq 0.95$) with moderate level of noise ($\lambda \leq 0.01$). 
Casapao \textit{et al.}'s other work \cite{11250296} estimates quantum states using an entanglement purification protocol called the double selection protocol \cite{PhysRevA.80.042308}, but they assume perfect local gates, whereas we deal with the case in which depolarizing errors are present at each gate. This work also helps us advance from a theoretical framework toward an operational system for our data center testbed.

Moreover, the Quantum Internet is at hand \cite{9951258, rfc9340, Wehner18:eaam9288}; experimentalists have succeeded in creating long-distance quantum entanglement \cite{PhysRevLett.130.213601, doi:10.1126/science.abg1919, doi:10.1126/science.aan3211, Pfaff29052014}. Fidelity estimation of the long-distance entanglement will be of great importance there. 
Although the full, global Quantum Internet has the capabilities to support distributed quantum computing \cite{Cuomo:2020aa}, the Quantum Internet has not yet been realized. By enabling a faster and more efficient delivery of high-fidelity Bell pairs, this work can support distributed quantum computing and realization of the Quantum Internet.

\section{\label{sec:protocol}Protocol}
In this section we describe our protocol for conducting both state reconstruction and purification of a Bell-diagonal state. First we introduce a set of assumptions made throughout the design of the protocol. We assume that each of the quantum gates involved in the circuit suffers from the depolarizing error of depolarizing parameter $\lambda$ and that the input states are all Bell-diagonal. Under these assumptions we conduct purification of the input states along with state reconstruction of the states. We first treat the case when $\lambda=0$, i.e., when there is no depolarizing error at any gate, and after that we treat the general case. \par
Purification is a process where a multiple number of lower-fidelity Bell pairs are consumed to obtain a single higher-fidelity Bell pair. In our method of purification we use the circuit in Figure \ref{fig:the_circuit_with_number}, where four identical low-fidelity Bell pairs are consumed to yield one higher-fidelity Bell pair. Purification is also a probabilistic process and we do not always succeed. Purification \textit{succeeds} with a certain probability $P$ and \textit{fails} with probability $1-P$. Purification success or failure is governed by the measurement outcomes at the measurement locations in the circuit. In our method of purification, purification succeeds if and only if in Figure \ref{fig:the_circuit_with_number} 1) the parity agrees at the measurement locations 1 and 2 and 2) the parity disagrees at the measurement location 3. \par

Our entire protocol is described in the form of pseudocode in Algorithm \ref{pseudocode_protocol}. We first estimate the fidelity using the circuit, and if the estimated fidelity is above the threshold (which we assume to be either $0.9$ or $0.99$ in the discussions below) then we use that Bell pair for further applications. If not, then we let the low-fidelity Bell pairs be passed on to the circuit and use it as the purification circuit to enhance the fidelity of them. We use the purification circuit recursively until the estimated fidelity reaches above the threshold. \par

\begin{algorithm}[H]
\caption{The pseudocode for our purification scheme: how it purifies the input states}
\label{pseudocode_protocol}
\begin{algorithmic}[1]
\Require{The given state $given\_state$ and the given threshold fidelity $given\_threshold$}
\Ensure{The resultant purified state $purified\_state$}
\State $state \gets given\_state$
\State $threshold \gets given\_threshold$
\State $F\_est \gets$ estimate\_actual\_fidelity(state)
\If{$F\_est < threshold$}
        \State $state \gets \text{purify\_the\_state}(state,threshold)$
\EndIf
\State return($state$)
\Function{purify\_the\_state}{$state, threshold$}
\State $F_{state} \gets \text{estimate\_actual\_fidelity}(state)$
\If{$F_{state}>threshold$}
    \State \Return $state$
\Else
    \State purify($state$)
    \State \Return purify\_the\_state($state$, $threshold$)
\EndIf
\EndFunction
\end{algorithmic}
\end{algorithm}

For the performance of the purification, we present the success probability of the purification and the resultant state upon a successful purification. Let 
\begin{equation}\label{eq:rho_0_general}
    \begin{aligned}
\rho_0=&a\ket{\Phi^+}\bra{\Phi^+}+b\ket{\Psi^+}\bra{\Psi^+}\\
+&c\ket{\Psi^-}\bra{\Psi^-}+d\ket{\Phi^-}\bra{\Phi^-}
    \end{aligned}
\end{equation} be the initial state and consider conducting purification to this state. We need to prepare 4 identical copies of the state and input them to the circuit. As we will see later, the probability that the parity of the measurement outcomes at the measurement location 1 agrees is $\widetilde{P_1}=(a+b)^2+(c+d)^2$. Similarly the probability that the parity of the measurement outcomes at the measurement location 2 agrees is $\widetilde{P_2}=(a+d)^2+(b+c)^2$. Finally, the probability that the parity of the measurement outcomes at the measurement location 3 \textit{disagrees} is, using the notation in equations \ref{eq:aa1} through \ref{eq:dd2}, $\widetilde{P_3}=(A_1+B_1)(A_2+C_2)+(C_1+D_1)(B_2+D_2)$. The purification success probability $\widetilde{P}$ is now  
\begin{equation*}
\begin{split}
\widetilde{P}=&\prod_{i=1}^3 \widetilde{P_i}\\=&a^4+a^2 \left(2 b^2+4 b d+(c+d)^2\right)+2 a \left(2 b^2 c+c^2 d+d^3\right)\\
&+b^4+b^2 (c+d)^2+2 b c \left(c^2+d^2\right)+2 c d \left(c^2+d^2\right)
\end{split}
\end{equation*} and the coefficients of the resultant state $\rho_1 =  a_1\ket{\Phi^+}\bra{\Phi^+}+b_1\ket{\Psi^+}\bra{\Psi^+}+c_1\ket{\Psi^-}\bra{\Psi^-}+d_1\ket{\Phi^-}\bra{\Phi^-}$ are as in Table \ref{tab:table1}. 

\begin{table}[]
\caption{\label{tab:table1}%
Coefficients of the Bell-diagonal state $\rho_1$, the state after the initial state $\rho_0$ has undergone purification. Each is expressed as a function of the old coefficients, $a$, $b$, $c$, and $d$.
}
\begin{tabular}{|c|c|}\hline
Coefficient & Value \\ \hline
$a_1$ &  $\frac{\begin{matrix} a^4+a^2 \left(b^2+d^2\right)+b^2 d (2 c+d)+2 c^3 d \end{matrix}}{\begin{matrix} a^4+a^2 \left(2 b^2+4 b d+(c+d)^2\right)+2 a \left(2 b^2 c+c^2 d+d^3\right)&\\+b^4+b^2 (c+d)^2+2 b c \left(c^2+d^2\right)+2 c d \left(c^2+d^2\right)\end{matrix}}$ \\ \hline
$b_1$ &  $\frac{\begin{matrix}a^2 \left(b^2+c (c+2 d)\right)+b^4+b^2 c^2+2 c d^3\end{matrix}}{\begin{matrix} a^4+a^2 \left(2 b^2+4 b d+(c+d)^2\right)+2 a \left(2 b^2 c+c^2 d+d^3\right)&\\+b^4+b^2 (c+d)^2+2 b c \left(c^2+d^2\right)+2 c d \left(c^2+d^2\right)\end{matrix}}$ \\ \hline
$c_1$ &  $\frac{\begin{matrix}2 a \left(2 b^2 c+c^2 d+d^3\right)\end{matrix}}{\begin{matrix} a^4+a^2 \left(2 b^2+4 b d+(c+d)^2\right)+2 a \left(2 b^2 c+c^2 d+d^3\right)&\\+b^4+b^2 (c+d)^2+2 b c \left(c^2+d^2\right)+2 c d \left(c^2+d^2\right)\end{matrix}}$ \\ \hline
$d_1$ &  $\frac{\begin{matrix}2 b (c^3 + 2 a^2 d + c d^2)\end{matrix}}{\begin{matrix} a^4+a^2 \left(2 b^2+4 b d+(c+d)^2\right)+2 a \left(2 b^2 c+c^2 d+d^3\right)&\\+b^4+b^2 (c+d)^2+2 b c \left(c^2+d^2\right)+2 c d \left(c^2+d^2\right)\end{matrix}}$ \\ \hline 
\end{tabular}
\end{table}
For example, if the initial state $\rho_0$ is the Werner state 
\begin{equation}\label{eq:rho_0}
\begin{aligned}
\rho_0=&0.7\ket{\Phi^+}\bra{\Phi^+}+0.1\ket{\Psi^+}\bra{\Psi^+}\\
+&0.1\ket{\Psi^-}\bra{\Psi^-}+0.1\ket{\Phi^-}\bra{\Phi^-},
\end{aligned}
\end{equation} then the purification success probability for the complete circuit in Figure \ref{fig:the_circuit_with_number} is $0.296$ and the coefficients of the resultant state $\rho_1$ are $(a_1,b_1,c_1,d_1)=(0.845946, 0.0675676, 0.0189189, 0.0675676)$. Employing the recursive scheme, the purification protocol can be recursively iterated, in which case the resultant fidelity asymptotically approaches $1$. For example, after the second iteration of the protocol the resultant state $\rho_2$ becomes 
\begin{align*}
    \rho_2 = &0.964084\ket{\Phi^+}\bra{\Phi^+}+0.0100128\ket{\Psi^+}\bra{\Psi^+}\\
    +&0.00158934\ket{\Psi^-}\bra{\Psi^-}+0.0243137\ket{\Phi^-}\bra{\Phi^-}.
\end{align*} The outcomes out of the further iterations of the purification of the Werner state $\rho_0$ are written in Table \ref{tab:tbl_further_puri}. \par

\begin{table}[]
\caption{\label{tab:tbl_further_puri}%
Coefficients of the state $\rho_m$ after the initial Werner state $\rho_0$ as in equation \ref{eq:rho_0}
underwent $m$ purifications, where the recursive scheme is adopted.
}
\begin{tabular}{|c|c|c|c|c|}
\hline
\multicolumn{1}{|l|}{} & $a$        & $b$            & $c$                     & $d$                      \\ \hline
$\rho_0$               & $0.7$ & $0.1$ & $0.1$ & $0.1$ \\ \hline
$\rho_1$               & $0.845946$ & $0.0675676$ & $0.0189189$ & $0.0675676$              \\ \hline
$\rho_2$               & $0.964084$ & $0.0100128$ & $0.00158934$ & $0.0243137$              \\ \hline
$\rho_3$               & $0.998728$ & $0.000193408$ & $0.0000328621$ & $0.0010456$             \\ \hline
$\rho_4$               & $0.999$ & $1.07\times10^{-7}$ & $2.30\times10^{-9}$ & $8.10\times10^{-7}$           \\ \hline
\end{tabular}
\end{table}

When the depolarizing error of depolarizing parameter $\lambda$ is present at each gate, the process of the purification above yields lower fidelity Bell pairs, with the fidelity depending on the value of $\lambda$. For example, consider the case when the original Bell pair is the Werner state $\rho_0$ as in equation \ref{eq:rho_0}
and the depolarizing parameter $\lambda=0.1$. In this case the purification in fact decreases the fidelity of the original Bell pair, where $1$ purification yields the fidelity $F=0.58$, $2$ purifications $F=0.46$, $3$ purifications $F=0.33$, $4$ purifications $F=0.26$, and $5$ purifications $F=0.25$. Therefore the purification is useless for the purpose of improving the fidelity of the original Bell pair in the case $\lambda=0.1$. When $\lambda=0.01$ the situation is better, where $1$ purification yields the fidelity $F=0.81$, $2$ purifications $F=0.92$, $3$ purifications $0.97$, $4$ purifications $0.98$, and $5$ purifications $F=0.98$. In this case the process of purification works, although the degree of improvement is decreased by the presence of the depolarizing error. It will be worth determining the threshold of the depolarizing parameter at which the process of purification actually improves the fidelity of the original Bell pair. We conducted numerical experiments to determine the threshold. It is about $\lambda = 0.053$ for the initial Werner state $F=0.7$ and it is about $\lambda = 0.006$ for the initial Werner state $F=0.99$.

\section{\label{sec:calc_of_rho_reconst}Calculation of $\rho_{\textrm{reconstructed}}$}

In this section, we explain in detail how the fidelity estimation is carried out. Again we first consider the case where the depolarizing parameter $\lambda$ is $0$, i.e., where there is no error at each gate, and proceed to the general case of $\lambda$ being arbitrary. Assume $\lambda=0$, let $\rho_0$ be the state as in equation \ref{eq:rho_0_general} and suppose we want to reconstruct this state.  
Throughout this section, we assume that the initial fidelity, $a$, is greater than or equal to $1/2$. We will collect the measurement outcomes in the Z, X, and Y bases and attempt to determine the coefficients $a$, $b$, $c$ and $d$. Following the circuit in Figure \ref{fig:the_circuit_with_number}, the possible outcomes and corresponding probabilities at the measurement location 1 are as in Table \ref{tab:TblOutcomeSet1}.

\begin{table}[]
\caption{Possible measurement outcomes at the measurement location 1 and their respective probabilities in terms of the coefficients of $\rho_0$.}
\centering
\begin{tabular}{|l|l|}
\hline
Outcomes at 1 & Probability                 \\ \hline
$00$         & $\left\{(a+b)^2+(c+d)^2\right\}/2$ \\ \hline
$01$         & $(a+b)(c+d)$ \\ \hline
$10$         & $(a+b)(c+d)$ \\ \hline
$11$         & $\left\{(a+b)^2+(c+d)^2\right\}/2$ \\ \hline
\end{tabular}
\label{tab:TblOutcomeSet1}
\end{table}

We define $p_1$ to be the proportion of the number of measurement outcomes whose parity agrees to the total number of measurement outcomes at the measurement location 1, i.e., the proportion of getting $00$ or $11$ at the measurement location 1. It is reasonable to expect that $p_1$ converges to $(a+b)^2+(c+d)^2$ as the total number of measurements increases. We use this property of convergence to infer the coefficients $a$, $b$, $c$ and $d$ from the measurement outcomes at the three measurement locations. In fact, let $x=a+b$ to obtain an equation of $x$ in terms of $p_1$:
\begin{equation}
p_1=x^2+(1-x)^2=2x^2-2x+1.
\end{equation} This implies that 
\begin{eqnarray}
a+b&=&\frac{1+\sqrt{2 p_1 -1}}{2}  \label{eq:eqrel1}\\
c+d&=&\frac{1-\sqrt{2 p_1 -1}}{2}  \label{eq:eqrel2}
\end{eqnarray}
 Now we have an estimated value of $a+b$ (and $c+d$) in terms of $p_1$.

Similarly, the possible outcomes and corresponding probabilities at the measurement location 2 are as in Table \ref{tab:TblOutcomeSet2}. Define $p_2$ to be the proportion of the number of measurement outcomes at the measurement location 2 whose parity agrees to the total number of measurement outcomes, i.e., the proportion of getting $00$ or $11$ at the measurement location 2. It is expected that $p_2$ converges to $(a+d)^2+(b+c)^2$ as the total number of measurements increases. In this case let $x=a+d$ to obtain an equation of $x$ in terms of $p_2$: 
\begin{equation}
p_2=x^2+(1-x)^2.
\end{equation} Assuming $a+d \geq 1/2$ (which is satisfied because the initial fidelity, that is, the value of $a$, is assumed to be $\geq 1/2$), solve this equation to obtain the relations
\begin{eqnarray}
  a+d &=& \frac{1+\sqrt{2p_2-1}}{2}, \label{eq:eqrel3} \\ 
  b+c &=& \frac{1-\sqrt{2p_2-1}}{2}. 
\end{eqnarray} Now we have an estimated value of $a+d$ (and $b+c$) in terms of $p_2$. 
\\

\begin{table}[]
\caption{Possible measurement outcomes at the measurement location 2 and their respective probabilities in terms of the coefficients of $\rho_0$.}
\centering
\begin{tabular}{|l|l|}
\hline
Outcomes at 2 & Probability                 \\ \hline
$00$         & $\left\{(a+d)^2+(b+c)^2\right\}/2$ \\ \hline
$01$         & $(a+d)(b+c)$                \\ \hline
$10$         & $(a+d)(b+c)$                \\ \hline
$11$         & $\left\{(a+d)^2+(b+c)^2\right\}/2$ \\ \hline
\end{tabular}
\label{tab:TblOutcomeSet2}
\end{table}
Finally, we define the new quantities $A_1$, $B_1$, $C_1$, $D_1$, $A_2$, $B_2$, $C_2$, and $D_2$ as follows:
\begin{eqnarray}
A_1 &=& a^2+d^2+ab+cd \label{eq:aa1}\\
B_1&=&b^2+c^2+ab+cd \label{eq:bb1}\\
C_1&=&a c+b d+2bc \label{eq:cc1}\\
D_1&=& ac+bd+2 a d\label{eq:dd1}\\
A_2&=&a^2+b^2+ad+bc \label{eq:aa2}\\
B_2&=&c^2+d^2+ad+bc \label{eq:bb2}\\
C_2&=&ac+bd+2cd\label{eq:cc2}\\
D_2&=&ac+bd+2ab\label{eq:dd2}.
\end{eqnarray}
Then the possible measurement outcomes at the measurement location 3 and their respective probabilities are as in Table \ref{tab:TblOutcomeSet3}.
\begin{table}[]
\caption{Possible measurement outcomes at the measurement location 3 and their respective probabilities in terms of the coefficients of $\rho_0$.}
\centering
\begin{tabular}{|l|l|}
\hline 
Outcome at 3  & Probability                 \\ & \\[-1em] \hline
$00$         & $\bigl\{(A_1+B_1)(B_2+D_2)+(C_1+D_1)(A_2+C_2)\bigr\}/2$ \\ & \\[-1em] \hline
$01$         & $\bigl\{(A_1+B_1)(A_2+C_2)+(C_1+D_1)(B_2+D_2)\bigr\}/2$ \\ & \\[-1em]\hline
$10$         & $\bigl\{(A_1+B_1)(A_2+C_2)+(C_1+D_1)(B_2+D_2)\bigr\}/2$ \\ & \\[-1em] \hline
$11$         & $\bigl\{(A_1+B_1)(B_2+D_2)+(C_1+D_1)(A_2+C_2)\bigr\}/2$ \\ \hline
\end{tabular}
\label{tab:TblOutcomeSet3}
\end{table} 
We define $p_3$ to be the proportion of the number of measurement outcomes at the measurement location 3 whose parity agrees to the total number of measurement outcomes, i.e., the proportion of getting $00$ or $11$ at the measurement location 3. Then we have
\begin{multline}
p_3 =p_1 (-4 a \sqrt{2 p_2-1}+\sqrt{2 p_1-1} \sqrt{2 p_2-1}\\+2 p_2+\sqrt{2 p_2-1}-1)+\frac{1}{2} (4 a \sqrt{2 p_2-1}\\-\sqrt{2 p_1-1} \sqrt{2 p_2-1}-2 p_2-\sqrt{2 p_2-1}+2)
\end{multline}
This implies that we can estimate the value of $a$ from $p_1$, $p_2$, and $p_3$ using the linear equation of $a$ above, and that the more measurements we conduct the more accurate the estimate becomes. Using the relation $a+b+c+d=1$ and the equations \ref{eq:eqrel1}, \ref{eq:eqrel2} and \ref{eq:eqrel3}, we can derive the value of $b$, $c$, and $d$ as well. \par

Next assume that the depolarizing parameter $\lambda$ is arbitrary. In this case the expressions involved in the state estimation become much more complicated, although the basic idea behind the process remains the same. We first introduce three functions $f_1(\lambda, a,b,c,d)$, $f_2(\lambda, a,b,c,d)$, and $f_3(\lambda, a,b,c,d)$ as follows: let $f_i(\lambda, a, b, c, d)$ be the probability that we measure $00$ or $11$ at the measurement location $i$ ($i=1,2,3$) in terms of $\lambda$, $a$, $b$, $c$ and $d$ and let $p_i$ be the proportion of actually having measured $00$ or $11$ to the total number of measurements at the measurement location $i$ ($i=1,2,3$). We solve the system of equations
\begin{equation}\label{eq:system_of_eqns_state_reconstruction}
\left\{ \begin{aligned} 
  f_1(\lambda, a,b,c,d)-p_1&=0\\
  f_2(\lambda, a,b,c,d)-p_2&=0\\
  f_3(\lambda, a,b,c,d)-p_3&=0\\
  a+b+c+d&=1
\end{aligned} \right. 
\end{equation} to compute the values of $a$, $b$, $c$, and $d$. It turns out that, for each fixed value of $\lambda=\lambda_0$, $f_1$ and $f_2$ are homogeneous polynomials of degree $2$ in $a$, $b$, $c$, and $d$, and $f_3$ is a homogeneous polynomial of degree $4$ in $a$, $b$, $c$ and $d$. The system of equations therefore consists of two quadratic equations, one quartic equation, and one linear equation. We can solve the system by appropriate substitution and obtain the values of $a$, $b$, $c$ and $d$. The values so obtained will be sufficiently close to the actual coefficients of the state in question if we take a large enough sample. The actual number of states needed for the values to be sufficiently close will be investigated in the next section. Unfortunately the concrete forms of these functions $f_i(\lambda,a,b,c,d)$ are in general difficult to calculate due to the complexity of the expressions involved unless the value of $\lambda$ is numerically given. Hence for inversion we use those $f_i$ with $\lambda$ substituted with a numerically explicit value. For $\lambda=0$ we have already seen the process of inversion earlier.

\section{\label{sec:simulation}Simulation and Results}
We have written a simulator in Python to numerically test the state reconstruction method described above. 
We conducted simulation experiments for the Werner state 
\begin{align*}
\rho_{\textrm{reality}}&=F\ket{\Phi^+}\bra{\Phi^+}+\frac{1-F}{3}\ket{\Psi^+}\bra{\Psi^+}\\ &+\frac{1-F}{3}\ket{\Psi^-}\bra{\Psi^-}+\frac{1-F}{3}\ket{\Phi^-}\bra{\Phi^-}
\end{align*} for three $F$'s: $F=0.95$, $F=0.7$, and $F=0.55$. We also experimented with two distinct depolarizing parameters, $\lambda=0$ and $\lambda=0.1$. We have thus $6$ distinct configurations in total in the set of experiments. For each configuration we performed the prescribed state reconstruction method with the number of circuit executions starting with $1,000$, and with increment of $1,000$, up to $70,000$, and for each such configuration we performed the same experiment $20$ times, obtaining $20$ samples for each configuration. We measured 1) the mean of the reconstruction fidelity (i.e., $F_{reconstructed}$ in Figure \ref{fig:OkaTriangle}) across these $20$ samples, 2) the standard deviation in the reconstruction fidelity across these $20$ samples, and 3) the mean of the trace distance between $\rho_{\textrm{reality}}$ and $\rho_{\textrm{reconstructed}}$ across these $20$ samples. These quantities are illustrated in Figures \ref{fig:mean_reconst_fidelity}, \ref{fig:stddev_reconst_fidelity} and \ref{fig:mean_trace_distance}. 

\begin{figure}
\includegraphics[keepaspectratio,scale=0.35]{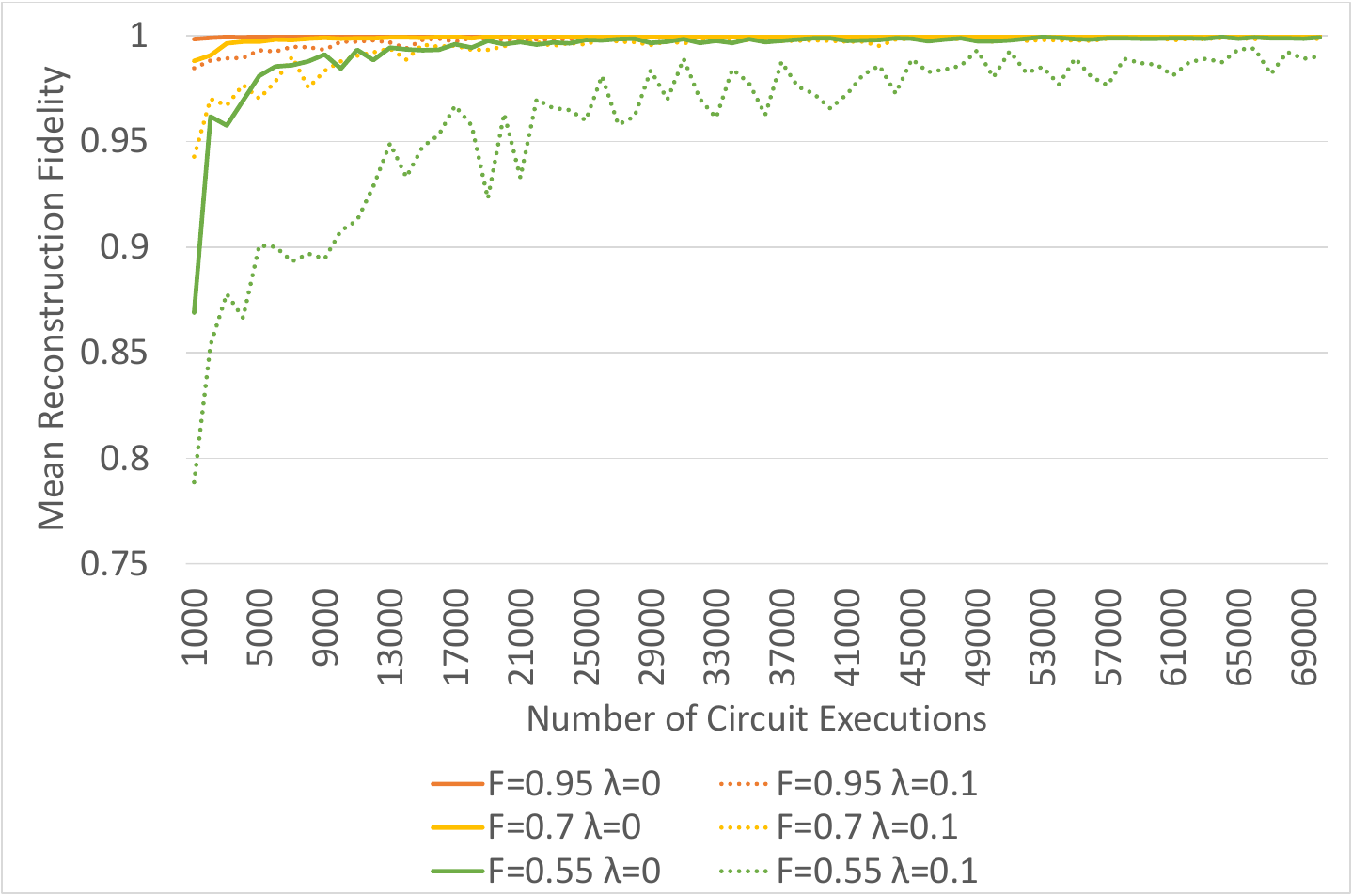}
\caption{\label{fig:mean_reconst_fidelity}The mean of the fidelity between $\rho_{\textrm{reality}}$ and $\rho_{\textrm{reconstructed}}$ for $F=0.95$, $F=0.7$ and $F=0.55$, each for the depolarizing parameter either $\lambda=0$ or $\lambda=0.1$ across the $20$ samples taken. The horizontal axis is the number of circuit executions and the vertical axis is the mean reconstruction fidelity.}
\end{figure}
\begin{figure}
\includegraphics[keepaspectratio,scale=0.35]{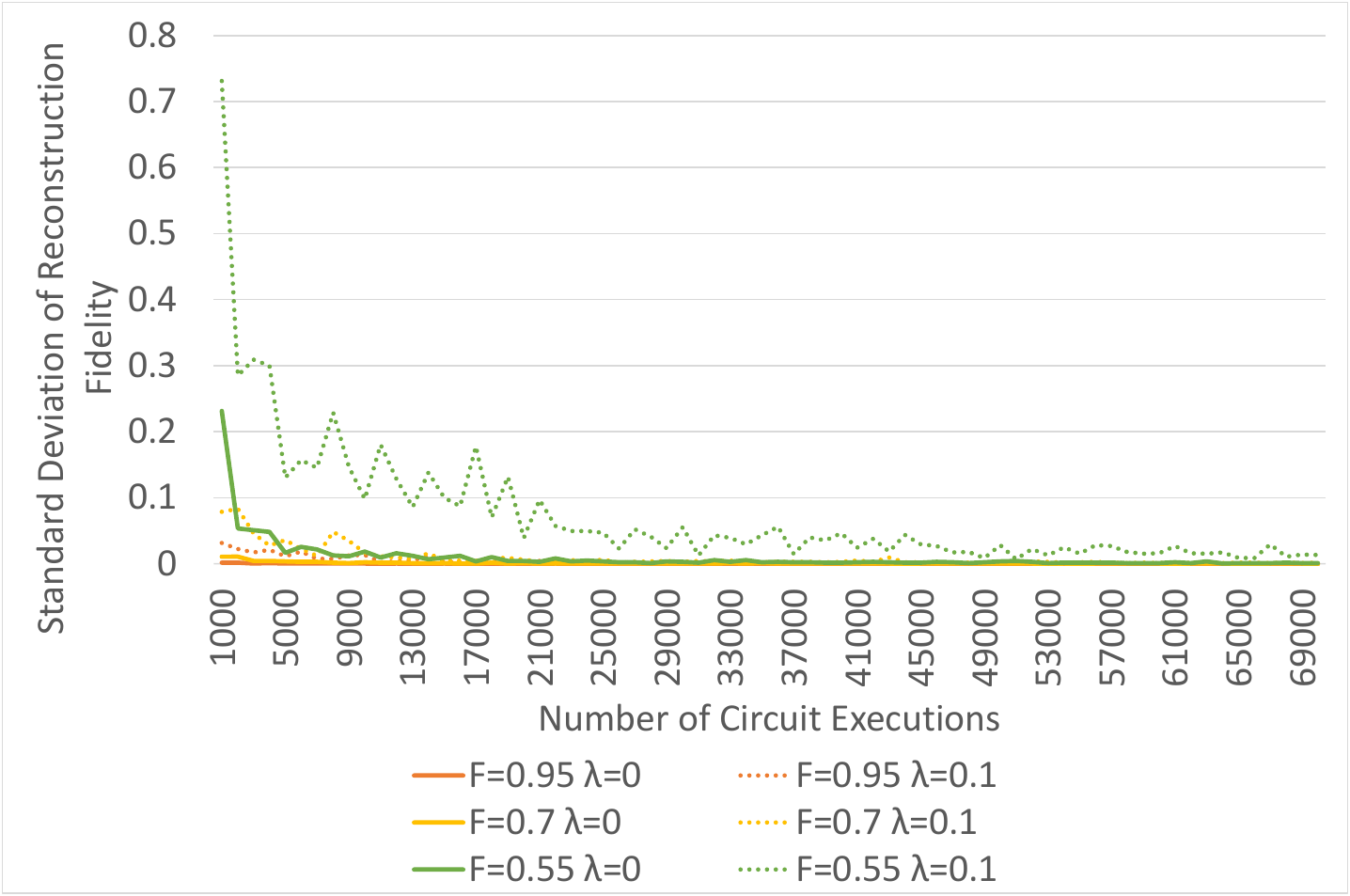}
\caption{\label{fig:stddev_reconst_fidelity}The standard deviation of the fidelity between $\rho_{\textrm{reality}}$ and $\rho_{\textrm{reconstructed}}$ for $F=0.95$, $F=0.7$ and $F=0.55$, each for the depolarizing parameter either $\lambda=0$ or $\lambda=0.1$ across the $20$ samples taken. The horizontal axis is the number of circuit executions and the vertical axis is the standard deviation of the reconstruction fidelity.}
\end{figure}
\begin{figure}
\includegraphics[keepaspectratio,scale=0.35]{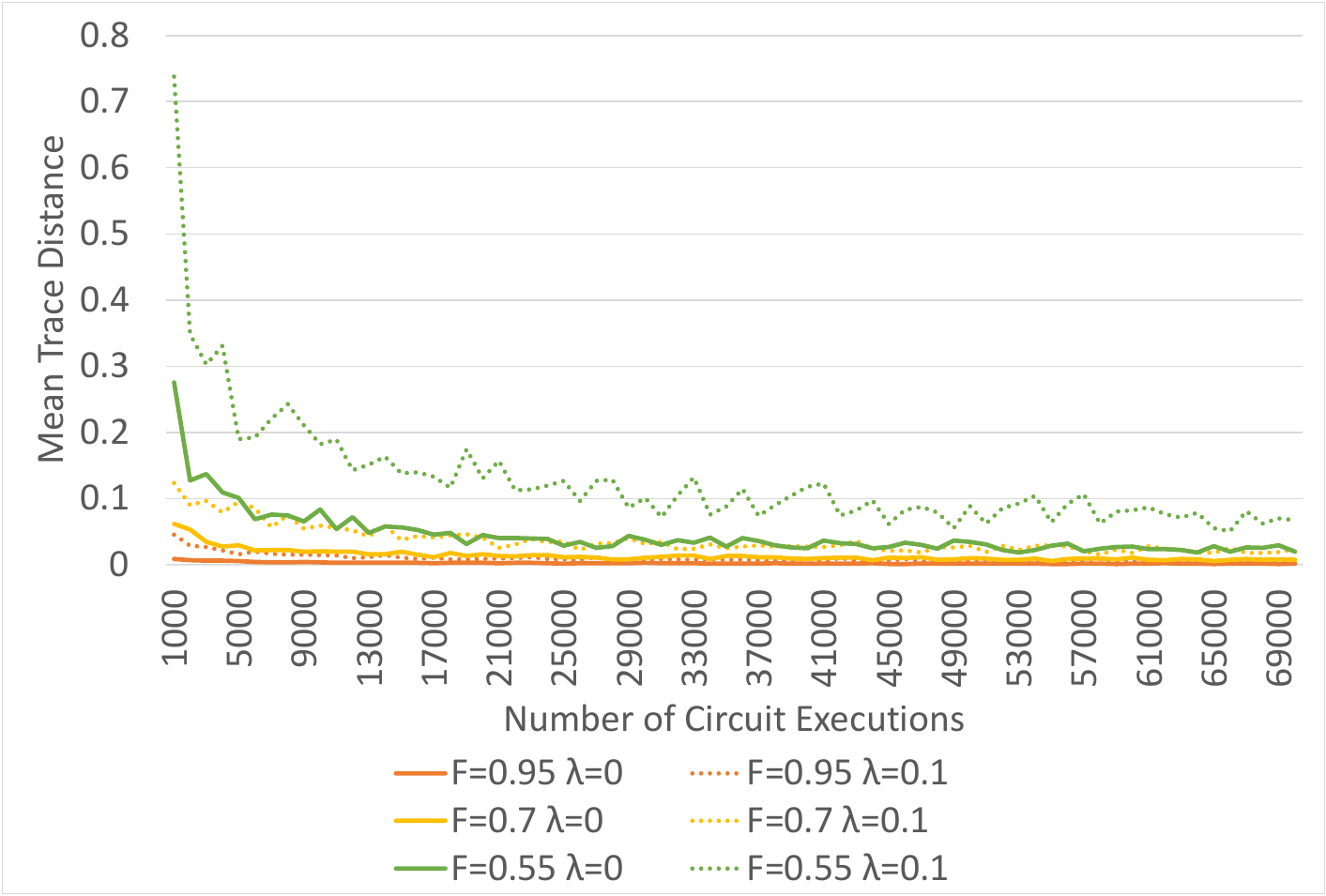}
\caption{\label{fig:mean_trace_distance}The mean of the trace distance between $\rho_{\textrm{reality}}$ and $\rho_{\textrm{reconstructed}}$ for $F=0.95$, $F=0.7$ and $F=0.55$, each for the depolarizing parameter either $\lambda=0$ or $\lambda=0.1$ across the $20$ samples taken. The horizontal axis is the number of circuit executions and the vertical axis is the mean trace distance.}
\end{figure}

Several observations can be made by looking at these data. First, the higher the fidelity of the base state is, the higher the reconstruction fidelity becomes and the smaller the standard deviation in the reconstruction fidelity becomes. This means that the reconstruction process becomes more and more accurate when the actual fidelity becomes higher. Second, when the depolarizing parameter is positive, the reconstruction fidelity drops and the standard deviation of the reconstruction fidelity becomes higher. This is due to the change in the reconstruction function $g$ to be introduced in Theorem \ref{theoremForAHat} later.

\par

In order to derive the estimation of the number of Bell pairs consumed to achieve a certain threshold fidelity, we wish to characterize the probability distribution that the estimation of the fidelity follows. So let $X_i$ be the binomial distribution $B(n,f_i(\lambda, a,b,c,d))$ with $i=1,2,3$, so that $X_i$ is the number of ``success`` outcomes out of $n$ trials with success probability $f_i(\lambda, a,b,c,d)$. Then we have, considering the $p_i$ as random variables, $p_i = X_i / n$, where the $p_i$ are as defined at the end of Section \ref{sec:calc_of_rho_reconst}. Let $\wh{a}(n, \lambda)$ be the reconstructed value of $a$ through solving the system of equations (\ref{eq:system_of_eqns_state_reconstruction}). Then symbolic calculations via computer show:
\begin{theorem}\label{theoremForAHat}
When $\lambda=0$ we have \begin{align*}
\wh{a}(n,0) = (&2000p_1p_2   +2000p_1\sqrt{p_1 - 0.5}\sqrt{p_2 - 0.5} \\ + &1414p_1\sqrt{p_2 - 0.5}  - 
1000p_1   -1000p_2\\ -& 1000p_3  - 1000\sqrt{p_1 - 0.5}\sqrt{p_2 - 0.5} \\ -& 707\sqrt{p_2 - 0.5} + 
1000)/ \\ &((5657p_1 - 2828)\sqrt{p_2 - 0.5}).
\end{align*}
When $\lambda = 0.1$ we have \begin{align*}
\wh{a}(n,0.1) = (&1503p_1p_2 + 1503p_1\sqrt{p_1 - 0.5}\sqrt{p_2 - 0.5} \\ + &
957p_1\sqrt{p_2 - 0.5}  - 752p_1  - 752p_2 \\ - &
1414p_3 - 752\sqrt{p_1 - 0.5}\sqrt{p_2 - 0.5} \\ - &
478\sqrt{p_2 - 0.5}  + 1083)/ \\ (&
3826p_1\sqrt{p_2 - 0.5}  -1913\sqrt{p_2 - 0.5})
\end{align*}
\end{theorem}

It follows that the random variable $\wh{a}(n,\lambda)$ lies in the field $\Q(p_1, p_2, p_3, \sqrt{p_1-0.5},\sqrt{p_2-0.5})$ when $\lambda=0$ or $\lambda = 0.1$. We conjecture that $\wh{a}(n,\lambda)$ lies in the field $\Q(p_1, p_2, p_3, \sqrt{p_1-0.5},\sqrt{p_2-0.5})$ for a general $\lambda$, though it could not be proved here due to complexity of the expressions involved. Because each $p_i$ is well approximated by the normal distribution $\mathcal{N}(f_i(\lambda, a, b, c, d),  f_i(\lambda, a,b,c,d)(1-f_i(\lambda, a, b, c, d)))$ when $n$ is large enough, it follows that $\wh{a}$ is approximated by a normal distribution. This is formally expressed in the next theorem: 
\begin{theorem}\label{normalTheorem}
Let $\mathcal{N}(\mu, \sigma^2)$ denote the normal distribution of mean $\mu$ and variance $\sigma^2$. The random variable $\widehat{a}(n,0)$ in Theorem \ref{theoremForAHat} is approximated by the normal distribution $\mathcal{N}(a, \frac{1}{n}(\nabla g (f)^T\Sigma\nabla g(f))$ where $g:\mathbb{R}^3\rightarrow \mathbb{R}$ is defined by 
\begin{align*}
    g(x_1,x_2,x_3)\defeq (&2000x_1x_2   +2000x_1\sqrt{x_1 - 0.5}\sqrt{x_2 - 0.5} \\ + &1414x_1\sqrt{x_2 - 0.5}  - 
1000x_1   -1000x_2\\ -& 1000x_3  - 1000\sqrt{x_1 - 0.5}\sqrt{x_2 - 0.5} \\ -& 707\sqrt{x_2 - 0.5} + 
1000)/ \\ &((5657x_1 - 2828)\sqrt{x_2 - 0.5}),
\end{align*}
$\Sigma$ is the covariance matrix of the random vector $p=(p_1, p_2, p_3)$, and $f=(f_1(a,b,c,d,0), f_2(a,b,c,d,0), f_3(a,b,c,d,0))$.\par
Similarly, $\widehat{a}(n,0.1)$ in Theorem \ref{theoremForAHat} is approximated by the normal distribution 
$\mathcal{N}(a, \frac{1}{n}(\nabla g^\prime (f^\prime)^T\Sigma\nabla g^\prime(f^\prime))$ where $g^\prime:\mathbb{R}^3\rightarrow \mathbb{R}$ is defined by 
\begin{align*}
    g^\prime(x_1,x_2,x_3)\defeq (&1503x_1x_2 + 1503p_1\sqrt{x_1 - 0.5}\sqrt{x_2 - 0.5} \\ + &
957x_1\sqrt{x_2 - 0.5}  - 752x_1  - 752x_2 \\ - &
1414x_3 - 752\sqrt{x_1 - 0.5}\sqrt{x_2 - 0.5} \\ - &
478\sqrt{x_2 - 0.5}  + 1083)/ \\ (&
3826x_1\sqrt{x_2 - 0.5}  -1913\sqrt{x_2 - 0.5}),
\end{align*}
$\Sigma$ is the covariance matrix of the random vector $p=(p_1, p_2, p_3)$, and $f^\prime=(f_1(a,b,c,d,0.1), f_2(a,b,c,d,0.1), f_3(a,b,c,d,0.1))$
\end{theorem}
\begin{proof}
We utilize the multivariate delta method. Observe that the map $g$ is smooth on the open region $U=\{(x_1, x_2, x_3)\in \mathbb{R}^3 \mid x_1 > 1/2, x_2 > 1/2\}$. Note also that the random vector $p=(p_1, p_2, p_3)$ converges in distribution to the multivariate normal distribution $\mathcal{N}(f,\frac{1}{n}\Sigma)$ by the multivariate central limit theorem. By the multivariate delta method we have 
\begin{equation}
    \widehat{a}(n,0) \xrightarrow{D} \mathcal{N}\left(a, \frac{1}{n}\nabla g (f)^T\Sigma\nabla g(f)\right)
\end{equation} because $g(f)=a$ and, since we have assumed that $a\geq 1/2$, we have $f_1(a,b,c,d,0) > 1/2$ and $f_2(a,b,c,d,0) > 1/2$ (see the expressions in Section \ref{sec:calc_of_rho_reconst}). An entirely identical line of reasoning applies to the case $\lambda=0.1$ replacing $g$ with $g^\prime$ and $f$ with $f^\prime$, showing that \begin{equation}
    \widehat{a}(n,0.1) \xrightarrow{D} \mathcal{N}\left(a, \frac{1}{n}\nabla g^\prime (f^\prime)^T\Sigma\nabla g^\prime(f^\prime)\right)
\end{equation} This completes the proof.
\end{proof}

Theorem \ref{normalTheorem} implies that in order to consider the random variable $\wh{a}(n,0)$ or $\wh{a}(n,0.1)$ it suffices to consider the corresponding normal distribution of mean $a$ and some computable standard deviation when $n$ is large. It is also clear from the theorem that the standard deviation converges to $0$ as $n$ approaches infinity. 

\par

We consider two scenarios and in each scenario derive the estimated number of Bell pairs consumed to achieve the scenario's goal in the subsequent paragraphs: one scenario is where we measure identical Bell pairs multiple times to estimate their fidelity using the circuit's ability to reconstruct the state and the other is where we are given a state of initial fidelity $F_0$ and we are supposed to estimate the fidelity of it and purify it multiple times at the same time using the proposed circuit to achieve a certain threshold fidelity $F$. We treat the first scenario first. \par

We exploit Theorem \ref{normalTheorem} to assume that $\widehat{a}$ is normally distributed and to compute the standard deviation of $\widehat{a}$ in terms of the Bell-diagonal coefficients of the true state $\rho$ (i.e., $a$, $b$, $c$, and $d$ in the expression $\rho=a\ket{\Phi^+}\bra{\Phi^+}+b\ket{\Psi^+}\bra{\Psi^+}+c\ket{\Psi^-}\bra{\Psi^-}+d\ket{\Phi^-}\bra{\Phi^-}$), the number of circuit runs $n$ and the depolarizing parameter $\lambda$. By the standard arguments of confidence interval in the normal distribution, in $99.7\%$ of probability we are certain that the fidelity estimate lies within the closed interval $[a-3\sigma,a+3\sigma]$ where $a$ is the mean (true fidelity) and $\sigma$ is the standard deviation. We deal with only Werner states, so that we have only one free parameter $a=F$, and we have $b=c=d=(1-F)/3$. Let $\sigma(n)$ be the standard deviation given a particular true state and depolarizing parameter. We computed $\sigma(n)$ for $12$ distinct configurations using the formula in Theorem \ref{normalTheorem}, namely  $F=0.99$, $F=0.95$, $F=0.7$, and $F=0.55$ for a depolarizing parameter being one of $\lambda=0$, $\lambda=0.01$ and $\lambda=0.1$. For example, when $F=0.99$ and $\lambda=0$ we have $\sigma(n)=0.0880\frac{1}{\sqrt{n}}$. The $\sigma(n)$'s for the other configurations are listed in Table \ref{tab:sigma_n_s}. From these equations we can derive how large $n$ needs to be to achieve a certain threshold standard deviation. Using this derivation we computed the number of Bell pairs consumed to estimate the fidelity within the range of $F \pm 10^{-2}$ with $99.7\%$ of confidence, for these $12$ distinct configurations. The numbers of Bell pairs are listed in Table \ref{tab:numberOfBellPairsConsumed}.

\begin{table}[]
\caption{The $\sigma(n)$ for the $12$ distinct configurations. It is observed that the higher the true fidelity is, the lower the standard deviation is. It is also observed that higher depolarizing parameters make the standard deviation larger. }\label{tab:sigma_n_s}
\centering
\begin{tabular}{|c|c|}
\hline
                        & $\sigma(n)$                 \\ \hline
$F=0.99$, $\lambda=0$   & $0.0880\frac{1}{\sqrt{n}}$ \\ \hline
$F=0.99$, $\lambda=0.01$   & $0.1532\frac{1}{\sqrt{n}}$ \\ \hline
$F=0.99$, $\lambda=0.1$   & $0.7111\frac{1}{\sqrt{n}}$ \\ \hline
$F=0.95$, $\lambda=0$   & $0.2108\frac{1}{\sqrt{n}}$  \\ \hline
$F=0.95$, $\lambda=0.01$ & $0.2573\frac{1}{\sqrt{n}}$  \\ \hline
$F=0.95$, $\lambda=0.1$ & $0.8585\frac{1}{\sqrt{n}}$  \\ \hline
$F=0.7$, $\lambda=0$    & $1.164\frac{1}{\sqrt{n}}$   \\ \hline
$F=0.7$, $\lambda=0.01$  & $1.2864\frac{1}{\sqrt{n}}$  \\ \hline
$F=0.7$, $\lambda=0.1$  & $3.3251\frac{1}{\sqrt{n}}$  \\ \hline
$F=0.55$, $\lambda=0$   & $3.9171\frac{1}{\sqrt{n}}$  \\ \hline
$F=0.55$, $\lambda=0.01$ & $4.3296\frac{1}{\sqrt{n}}$ \\ \hline
$F=0.55$, $\lambda=0.1$ & $11.2091\frac{1}{\sqrt{n}}$ \\ \hline
\end{tabular}
\end{table}

\begin{table}[]
\caption{The number of Bell pairs consumed to be $99.7\%$ certain that $\widehat{a}$ lies in $[F-0.01,F+0.01]$.}\label{tab:numberOfBellPairsConsumed}
\begin{tabular}{|c|c|}
\hline
                        & \makecell{Number of Bell pairs needed to be $99.7\%$ \\ certain that $\widehat{a}$ lies in $[F-0.01, F+0.01]$} \\ \hline
$F=0.99$, $\lambda=0$   & $2,841$                                                                                     \\ \hline
$F=0.99$, $\lambda=0.01$   & $8,621$                                                                                     \\ \hline
$F=0.99$, $\lambda=0.1$   & $185,709$                                                                                     \\ \hline
$F=0.95$, $\lambda=0$   & $15,997$                                                                                     \\ \hline
$F=0.95$, $\lambda=0.01$   & $24,317$                                                                                     \\ \hline
$F=0.95$, $\lambda=0.1$ & $265,328$                                                                                    \\ \hline
$F=0.7$, $\lambda=0$    & $487,763$                                                                                    \\ \hline

$F=0.7$, $\lambda=0.01$  & $607,833$                                                                          \\ \hline
$F=0.7$, $\lambda=0.1$  & $3.98026\cdot 10^6$                                                                          \\ \hline
$F=0.55$, $\lambda=0$   & $5.52372\cdot 10^6$                                                                          \\ \hline
$F=0.55$, $\lambda=0.01$ & $6.88538\cdot 10^6$                                                                          \\ \hline
$F=0.55$, $\lambda=0.1$ & $4.52318\cdot 10^7$                                                                          \\ \hline
\end{tabular}
\end{table}

We go on to the other scenario, where we derive the estimated number of Bell pairs consumed to yield Bell pairs of fidelity $>0.9$ or $>0.99$ starting with the initial Werner state of fidelity $0.7$ and purifying it multiple times. By Theorem \ref{normalTheorem}, we may assume that the fidelity estimation $\widehat{a}$ is normally distributed. By the standard arguments of confidence interval, in $99.7\%$ of probability we are certain that the fidelity estimate lies within the closed interval $[a-3\sigma,a+3\sigma]$ where $a$ is the mean and $\sigma$ is the standard deviation. $\sigma$ is given by the theorem, and we computed it following the formula given in the theorem.

We first treat the case of generating Bell pairs of fidelity $>0.9$. According to Table \ref{tab:tbl_further_puri}, we need $2$ rounds of purification in order to achieve fidelity greater than $0.9$. According to Table \ref{tab:tbl_further_puri} we have $a=0.964$.  In order for the target fidelity $0.9$ to be separated from the range $[a-3\sigma,a+3\sigma]$, we need to have $\sigma \leq (0.964-0.9)/3=0.0213$. According to the standard deviation we computed, to have this inequality we need $66$ circuit runs, using the recursive scheme as the purification scheme. The purification success probability of the Werner state $\rho_0=0.7\ket{\Phi^+}\bra{\Phi^+}+0.1\ket{\Psi^+}\bra{\Psi^+}+0.1\ket{\Psi^-}\bra{\Psi^-}+0.1\ket{\Phi^-}\bra{\Phi^-}$ is $0.296$. The purification success probability of this state after the first purification (that is, purification success probability of the state $\rho_1$ in the notation of Table \ref{tab:tbl_further_puri}) is $0.538009$.

Combining all these results, we see the combined effect of purification and state reconstruction. Although it only takes two rounds of the circuit in Fig.~\ref{fig:the_circuit_with_number} to go from $F=0.7$ to $F=0.96$ \emph{when we know the full behavior of the system}, in order to learn that behavior requires more data. We estimate the number of Bell pairs to be consumed in order to prove that we are generating Bell pairs of fidelity higher than $0.9$, starting with the Werner state of fidelity $0.7$, to be $66\cdot 4^2 \cdot (1/0.296) \cdot (1/0.538009)=6,631$.

A similar reasoning applies to the case of generating Bell pairs of fidelity $>0.99$. According to Table \ref{tab:tbl_further_puri}, we need $3$ rounds of purification in order to achieve fidelity greater than $0.99$. According to Table \ref{tab:tbl_further_puri} we have $a=0.999$.  In order for the output fidelity $0.99$ to be statistically separated from the range $[a-3\sigma,a+3\sigma]$, we need to have $\sigma \leq (0.999-0.99)/3=0.003$. According to our calculation, to have this inequality we need $112$ circuit runs. The purification success probability of this state after the second purification is $0.865637$. From all of these results we deduce that the number of Bell pairs to be consumed in order to generate Bell pairs of fidelity higher than $0.99$, starting with the Werner state of fidelity $0.7$, is $112\cdot 4^3 \cdot (1/0.296) \cdot (1/0.538009)\cdot (1/0.865637)=51,997$. 

\par

We also performed simulation experiments using Stim. While when the local gate error rate is $0.1$, the suppression of errors is outweighed by the local gate errors, in the remaining cases the error rate constantly drops because of the purification circuit, as can be seen in Figure \ref{fig:simulation_using_Stim}, where the infidelity is defined as $1-\text{fidelity}$. It can be seen that the purification using this circuit is successful, effectively suppressing errors. Note that there are several data points that are not plotted for reasons to be explained now. The reason that the point of the purple line after 3 rounds of purification is missing is that the error rate becomes 0 and since the plot is logarithmic, the point goes to the negative infinity. The reason that the points of other than the purple line after 3 rounds of purification are missing is that the number of the shots that have been post-selected is zero and the division by zero is happening here.
\begin{figure}
\includegraphics[keepaspectratio,scale=0.37]{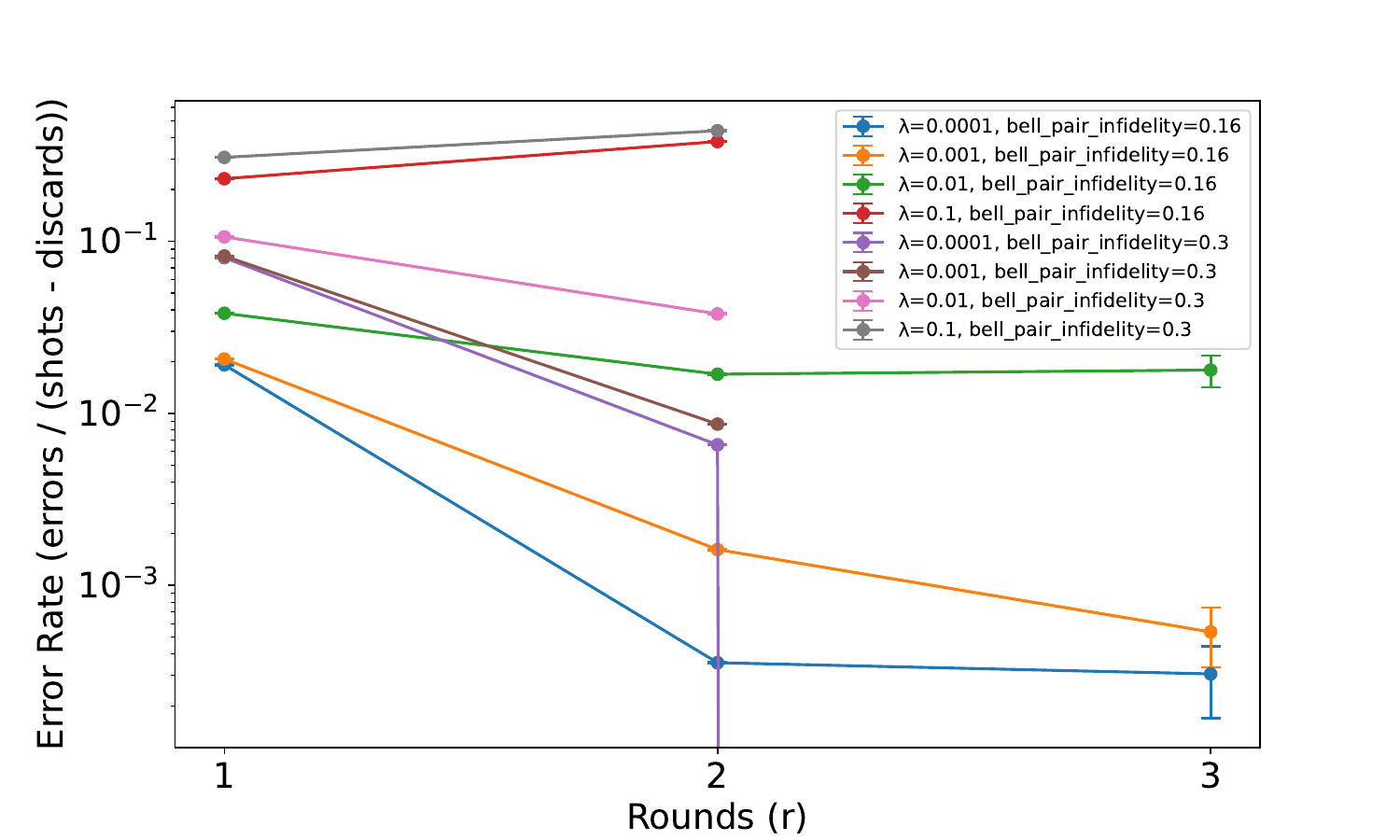}
\caption{\label{fig:simulation_using_Stim}The simulation results using the quantum simulator Stim. The horizontal axis is the number of purifications (rounds) and the vertical axis is the error rate. $\lambda$ represents the local error rate and bell\_pair\_infidelity is the infidelity of the initial Bell pair. Besides the case of local gate error rate being $0.1$, the error rate constantly drops because of the purification circuit.}
\end{figure}

\section{\label{sec:discussion}Discussion}

\par
We proposed a quantum circuit that is capable of both state reconstruction of Bell-diagonal states and purification (of Bell-diagonal states). We showed that as the number of measurement increases the accuracy of the state reconstruction also increases. We also calculated the estimated number of Bell pairs to conclude the fidelity estimate lies within a certain interval with $99.7\%$ confidence. Further, we showed that our circuit as a purification circuit makes the fidelity of the target state asymptotically approach $1$ when the local depolarizing gate error rate $\lambda$ is $0$. 
\par

When compared to Flammia \textit{et al.}'s work of fidelity estimation \cite{PhysRevLett.106.230501}, it can be said that the above method of state reconstruction gives some redundant information on the error terms (the coefficients of $\ket{\Psi^+}\bra{\Psi^+}$, $\ket{\Psi^-}\bra{\Psi^-}$, and $\ket{\Phi^-}\bra{\Phi^-}$) since it fully reconstructs Bell-diagonal states. \par
There is work on estimating various properties of a quantum state, including fidelity estimation, by Huang \textit{et al.} \cite{Huang:2020aa}, from very few quantum measurements. However, their emphasis is on the scaling with the system size and not on advantages over quantum state tomography given a fixed, small system size, as in our case. \par

\section{\label{sec:acknowledgments}Acknowledgments}
We thank David Elkouss and Joshua Carlo A. Casapao for fruitful discussions. This work was supported by JST SPRING, Japan Grant Number JPMJSP2123 and by JST Moonshot R$\&$D Grants JPMJMS2061 and JPMJMS226C.

\section{\label{sec:competing_interests}Competing Interests}
The authors declare no competing interests.

\bibliography{Oka_PRA_draft}

\providecommand{\noopsort}[1]{}\providecommand{\singleletter}[1]{#1}%
\begin{thebibliography}{33}%
\makeatletter
\providecommand \@ifxundefined [1]{%
 \@ifx{#1\undefined}
}%
\providecommand \@ifnum [1]{%
 \ifnum #1\expandafter \@firstoftwo
 \else \expandafter \@secondoftwo
 \fi
}%
\providecommand \@ifx [1]{%
 \ifx #1\expandafter \@firstoftwo
 \else \expandafter \@secondoftwo
 \fi
}%
\providecommand \natexlab [1]{#1}%
\providecommand \enquote  [1]{``#1''}%
\providecommand \bibnamefont  [1]{#1}%
\providecommand \bibfnamefont [1]{#1}%
\providecommand \citenamefont [1]{#1}%
\providecommand \href@noop [0]{\@secondoftwo}%
\providecommand \href [0]{\begingroup \@sanitize@url \@href}%
\providecommand \@href[1]{\@@startlink{#1}\@@href}%
\providecommand \@@href[1]{\endgroup#1\@@endlink}%
\providecommand \@sanitize@url [0]{\catcode `\\12\catcode `\$12\catcode `\&12\catcode `\#12\catcode `\^12\catcode `\_12\catcode `\%12\relax}%
\providecommand \@@startlink[1]{}%
\providecommand \@@endlink[0]{}%
\providecommand \url  [0]{\begingroup\@sanitize@url \@url }%
\providecommand \@url [1]{\endgroup\@href {#1}{\urlprefix }}%
\providecommand \urlprefix  [0]{URL }%
\providecommand \Eprint [0]{\href }%
\providecommand \doibase [0]{https://doi.org/}%
\providecommand \selectlanguage [0]{\@gobble}%
\providecommand \bibinfo  [0]{\@secondoftwo}%
\providecommand \bibfield  [0]{\@secondoftwo}%
\providecommand \translation [1]{[#1]}%
\providecommand \BibitemOpen [0]{}%
\providecommand \bibitemStop [0]{}%
\providecommand \bibitemNoStop [0]{.\EOS\space}%
\providecommand \EOS [0]{\spacefactor3000\relax}%
\providecommand \BibitemShut  [1]{\csname bibitem#1\endcsname}%
\let\auto@bib@innerbib\@empty
\bibitem [{\citenamefont {Barral}\ \emph {et~al.}(2025)\citenamefont {Barral}, \citenamefont {Cardama}, \citenamefont {D{\'\i}az-Camacho}, \citenamefont {Fa{\'\i}lde}, \citenamefont {Llovo}, \citenamefont {Mussa-Juane}, \citenamefont {V{\'a}zquez-P{\'e}rez}, \citenamefont {Villasuso}, \citenamefont {Pi{\~n}eiro}, \citenamefont {Costas}, \citenamefont {Pichel}, \citenamefont {Pena},\ and\ \citenamefont {G{\'o}mez}}]{BARRAL2025100747}%
  \BibitemOpen
  \bibfield  {author} {\bibinfo {author} {\bibfnamefont {D.}~\bibnamefont {Barral}}, \bibinfo {author} {\bibfnamefont {F.~J.}\ \bibnamefont {Cardama}}, \bibinfo {author} {\bibfnamefont {G.}~\bibnamefont {D{\'\i}az-Camacho}}, \bibinfo {author} {\bibfnamefont {D.}~\bibnamefont {Fa{\'\i}lde}}, \bibinfo {author} {\bibfnamefont {I.~F.}\ \bibnamefont {Llovo}}, \bibinfo {author} {\bibfnamefont {M.}~\bibnamefont {Mussa-Juane}}, \bibinfo {author} {\bibfnamefont {J.}~\bibnamefont {V{\'a}zquez-P{\'e}rez}}, \bibinfo {author} {\bibfnamefont {J.}~\bibnamefont {Villasuso}}, \bibinfo {author} {\bibfnamefont {C.}~\bibnamefont {Pi{\~n}eiro}}, \bibinfo {author} {\bibfnamefont {N.}~\bibnamefont {Costas}}, \bibinfo {author} {\bibfnamefont {J.~C.}\ \bibnamefont {Pichel}}, \bibinfo {author} {\bibfnamefont {T.~F.}\ \bibnamefont {Pena}},\ and\ \bibinfo {author} {\bibfnamefont {A.}~\bibnamefont {G{\'o}mez}},\ }\bibfield  {title} {\bibinfo {title} {{Review of Distributed Quantum Computing: From single QPU to High Performance Quantum
  Computing}},\ }\href {https://doi.org/https://doi.org/10.1016/j.cosrev.2025.100747} {\bibfield  {journal} {\bibinfo  {journal} {Computer Science Review}\ }\textbf {\bibinfo {volume} {57}},\ \bibinfo {pages} {100747} (\bibinfo {year} {2025})}\BibitemShut {NoStop}%
\bibitem [{\citenamefont {Awschalom}\ \emph {et~al.}(2021)\citenamefont {Awschalom}, \citenamefont {Berggren}, \citenamefont {Bernien}, \citenamefont {Bhave}, \citenamefont {Carr}, \citenamefont {Davids}, \citenamefont {Economou}, \citenamefont {Englund}, \citenamefont {Faraon}, \citenamefont {Fejer}, \citenamefont {Guha}, \citenamefont {Gustafsson}, \citenamefont {Hu}, \citenamefont {Jiang}, \citenamefont {Kim}, \citenamefont {Korzh}, \citenamefont {Kumar}, \citenamefont {Kwiat}, \citenamefont {Lon\ifmmode~\check{c}\else \v{c}\fi{}ar}, \citenamefont {Lukin}, \citenamefont {Miller}, \citenamefont {Monroe}, \citenamefont {Nam}, \citenamefont {Narang}, \citenamefont {Orcutt}, \citenamefont {Raymer}, \citenamefont {Safavi-Naeini}, \citenamefont {Spiropulu}, \citenamefont {Srinivasan}, \citenamefont {Sun}, \citenamefont {Vu\ifmmode \check{c}\else \v{c}\fi{}kovi\ifmmode~\acute{c}\else \'{c}\fi{}}, \citenamefont {Waks}, \citenamefont {Walsworth}, \citenamefont {Weiner},\ and\ \citenamefont
  {Zhang}}]{PRXQuantum.2.017002}%
  \BibitemOpen
  \bibfield  {author} {\bibinfo {author} {\bibfnamefont {D.}~\bibnamefont {Awschalom}}, \bibinfo {author} {\bibfnamefont {K.~K.}\ \bibnamefont {Berggren}}, \bibinfo {author} {\bibfnamefont {H.}~\bibnamefont {Bernien}}, \bibinfo {author} {\bibfnamefont {S.}~\bibnamefont {Bhave}}, \bibinfo {author} {\bibfnamefont {L.~D.}\ \bibnamefont {Carr}}, \bibinfo {author} {\bibfnamefont {P.}~\bibnamefont {Davids}}, \bibinfo {author} {\bibfnamefont {S.~E.}\ \bibnamefont {Economou}}, \bibinfo {author} {\bibfnamefont {D.}~\bibnamefont {Englund}}, \bibinfo {author} {\bibfnamefont {A.}~\bibnamefont {Faraon}}, \bibinfo {author} {\bibfnamefont {M.}~\bibnamefont {Fejer}}, \bibinfo {author} {\bibfnamefont {S.}~\bibnamefont {Guha}}, \bibinfo {author} {\bibfnamefont {M.~V.}\ \bibnamefont {Gustafsson}}, \bibinfo {author} {\bibfnamefont {E.}~\bibnamefont {Hu}}, \bibinfo {author} {\bibfnamefont {L.}~\bibnamefont {Jiang}}, \bibinfo {author} {\bibfnamefont {J.}~\bibnamefont {Kim}}, \bibinfo {author} {\bibfnamefont {B.}~\bibnamefont
  {Korzh}}, \bibinfo {author} {\bibfnamefont {P.}~\bibnamefont {Kumar}}, \bibinfo {author} {\bibfnamefont {P.~G.}\ \bibnamefont {Kwiat}}, \bibinfo {author} {\bibfnamefont {M.}~\bibnamefont {Lon\ifmmode~\check{c}\else \v{c}\fi{}ar}}, \bibinfo {author} {\bibfnamefont {M.~D.}\ \bibnamefont {Lukin}}, \bibinfo {author} {\bibfnamefont {D.~A.}\ \bibnamefont {Miller}}, \bibinfo {author} {\bibfnamefont {C.}~\bibnamefont {Monroe}}, \bibinfo {author} {\bibfnamefont {S.~W.}\ \bibnamefont {Nam}}, \bibinfo {author} {\bibfnamefont {P.}~\bibnamefont {Narang}}, \bibinfo {author} {\bibfnamefont {J.~S.}\ \bibnamefont {Orcutt}}, \bibinfo {author} {\bibfnamefont {M.~G.}\ \bibnamefont {Raymer}}, \bibinfo {author} {\bibfnamefont {A.~H.}\ \bibnamefont {Safavi-Naeini}}, \bibinfo {author} {\bibfnamefont {M.}~\bibnamefont {Spiropulu}}, \bibinfo {author} {\bibfnamefont {K.}~\bibnamefont {Srinivasan}}, \bibinfo {author} {\bibfnamefont {S.}~\bibnamefont {Sun}}, \bibinfo {author} {\bibfnamefont {J.}~\bibnamefont {Vu\ifmmode \check{c}\else
  \v{c}\fi{}kovi\ifmmode~\acute{c}\else \'{c}\fi{}}}, \bibinfo {author} {\bibfnamefont {E.}~\bibnamefont {Waks}}, \bibinfo {author} {\bibfnamefont {R.}~\bibnamefont {Walsworth}}, \bibinfo {author} {\bibfnamefont {A.~M.}\ \bibnamefont {Weiner}},\ and\ \bibinfo {author} {\bibfnamefont {Z.}~\bibnamefont {Zhang}},\ }\bibfield  {title} {\bibinfo {title} {{Development of Quantum Interconnects (QuICs) for Next-Generation Information Technologies}},\ }\href {https://doi.org/10.1103/PRXQuantum.2.017002} {\bibfield  {journal} {\bibinfo  {journal} {PRX Quantum}\ }\textbf {\bibinfo {volume} {2}},\ \bibinfo {pages} {017002} (\bibinfo {year} {2021})}\BibitemShut {NoStop}%
\bibitem [{\citenamefont {Bennett}\ \emph {et~al.}(1996{\natexlab{a}})\citenamefont {Bennett}, \citenamefont {Brassard}, \citenamefont {Popescu}, \citenamefont {Schumacher}, \citenamefont {Smolin},\ and\ \citenamefont {Wootters}}]{PhysRevLett.76.722}%
  \BibitemOpen
  \bibfield  {author} {\bibinfo {author} {\bibfnamefont {C.~H.}\ \bibnamefont {Bennett}}, \bibinfo {author} {\bibfnamefont {G.}~\bibnamefont {Brassard}}, \bibinfo {author} {\bibfnamefont {S.}~\bibnamefont {Popescu}}, \bibinfo {author} {\bibfnamefont {B.}~\bibnamefont {Schumacher}}, \bibinfo {author} {\bibfnamefont {J.~A.}\ \bibnamefont {Smolin}},\ and\ \bibinfo {author} {\bibfnamefont {W.~K.}\ \bibnamefont {Wootters}},\ }\bibfield  {title} {\bibinfo {title} {Purification of {Noisy Entanglement and Faithful Teleportation via Noisy Channels}},\ }\href {https://doi.org/10.1103/PhysRevLett.76.722} {\bibfield  {journal} {\bibinfo  {journal} {Phys. Rev. Lett.}\ }\textbf {\bibinfo {volume} {76}},\ \bibinfo {pages} {722} (\bibinfo {year} {1996}{\natexlab{a}})}\BibitemShut {NoStop}%
\bibitem [{\citenamefont {Bennett}\ \emph {et~al.}(1996{\natexlab{b}})\citenamefont {Bennett}, \citenamefont {DiVincenzo}, \citenamefont {Smolin},\ and\ \citenamefont {Wootters}}]{PhysRevA.54.3824}%
  \BibitemOpen
  \bibfield  {author} {\bibinfo {author} {\bibfnamefont {C.~H.}\ \bibnamefont {Bennett}}, \bibinfo {author} {\bibfnamefont {D.~P.}\ \bibnamefont {DiVincenzo}}, \bibinfo {author} {\bibfnamefont {J.~A.}\ \bibnamefont {Smolin}},\ and\ \bibinfo {author} {\bibfnamefont {W.~K.}\ \bibnamefont {Wootters}},\ }\bibfield  {title} {\bibinfo {title} {Mixed-state entanglement and quantum error correction},\ }\href {https://doi.org/10.1103/PhysRevA.54.3824} {\bibfield  {journal} {\bibinfo  {journal} {Phys. Rev. A}\ }\textbf {\bibinfo {volume} {54}},\ \bibinfo {pages} {3824} (\bibinfo {year} {1996}{\natexlab{b}})}\BibitemShut {NoStop}%
\bibitem [{\citenamefont {Fujii}\ and\ \citenamefont {Yamamoto}(2009)}]{PhysRevA.80.042308}%
  \BibitemOpen
  \bibfield  {author} {\bibinfo {author} {\bibfnamefont {K.}~\bibnamefont {Fujii}}\ and\ \bibinfo {author} {\bibfnamefont {K.}~\bibnamefont {Yamamoto}},\ }\bibfield  {title} {\bibinfo {title} {Entanglement purification with double selection},\ }\href {https://doi.org/10.1103/PhysRevA.80.042308} {\bibfield  {journal} {\bibinfo  {journal} {Phys. Rev. A}\ }\textbf {\bibinfo {volume} {80}},\ \bibinfo {pages} {042308} (\bibinfo {year} {2009})}\BibitemShut {NoStop}%
\bibitem [{\citenamefont {Goodenough}\ \emph {et~al.}(2024)\citenamefont {Goodenough}, \citenamefont {de~Bone}, \citenamefont {Addala}, \citenamefont {Krastanov}, \citenamefont {Jansen}, \citenamefont {Gijswijt},\ and\ \citenamefont {Elkouss}}]{10.1109/JSAC.2024.3380094}%
  \BibitemOpen
  \bibfield  {author} {\bibinfo {author} {\bibfnamefont {K.}~\bibnamefont {Goodenough}}, \bibinfo {author} {\bibfnamefont {S.}~\bibnamefont {de~Bone}}, \bibinfo {author} {\bibfnamefont {V.}~\bibnamefont {Addala}}, \bibinfo {author} {\bibfnamefont {S.}~\bibnamefont {Krastanov}}, \bibinfo {author} {\bibfnamefont {S.}~\bibnamefont {Jansen}}, \bibinfo {author} {\bibfnamefont {D.}~\bibnamefont {Gijswijt}},\ and\ \bibinfo {author} {\bibfnamefont {D.}~\bibnamefont {Elkouss}},\ }\bibfield  {title} {\bibinfo {title} {Near-term n to k distillation protocols using graph codes},\ }\href {https://doi.org/10.1109/JSAC.2024.3380094} {\bibfield  {journal} {\bibinfo  {journal} {IEEE J.Sel. A. Commun.}\ }\textbf {\bibinfo {volume} {42}},\ \bibinfo {pages} {1830–1849} (\bibinfo {year} {2024})}\BibitemShut {NoStop}%
\bibitem [{\citenamefont {Kim}\ \emph {et~al.}(2024)\citenamefont {Kim}, \citenamefont {Yun},\ and\ \citenamefont {Bae}}]{Kim_2025}%
  \BibitemOpen
  \bibfield  {author} {\bibinfo {author} {\bibfnamefont {J.}~\bibnamefont {Kim}}, \bibinfo {author} {\bibfnamefont {J.}~\bibnamefont {Yun}},\ and\ \bibinfo {author} {\bibfnamefont {J.}~\bibnamefont {Bae}},\ }\bibfield  {title} {\bibinfo {title} {Purification of noisy measurements and faithful distillation of entanglement},\ }\href {https://doi.org/10.1088/1751-8121/ada0fa} {\bibfield  {journal} {\bibinfo  {journal} {Journal of Physics A: Mathematical and Theoretical}\ }\textbf {\bibinfo {volume} {58}},\ \bibinfo {pages} {03LT01} (\bibinfo {year} {2024})}\BibitemShut {NoStop}%
\bibitem [{\citenamefont {Pan}\ \emph {et~al.}(2001)\citenamefont {Pan}, \citenamefont {Simon}, \citenamefont {Brukner},\ and\ \citenamefont {Zeilinger}}]{Pan2001}%
  \BibitemOpen
  \bibfield  {author} {\bibinfo {author} {\bibfnamefont {J.-W.}\ \bibnamefont {Pan}}, \bibinfo {author} {\bibfnamefont {C.}~\bibnamefont {Simon}}, \bibinfo {author} {\bibfnamefont {{\v{C}}.}~\bibnamefont {Brukner}},\ and\ \bibinfo {author} {\bibfnamefont {A.}~\bibnamefont {Zeilinger}},\ }\bibfield  {title} {\bibinfo {title} {Entanglement purification for quantum communication},\ }\href {https://doi.org/10.1038/35074041} {\bibfield  {journal} {\bibinfo  {journal} {Nature}\ }\textbf {\bibinfo {volume} {410}},\ \bibinfo {pages} {1067} (\bibinfo {year} {2001})}\BibitemShut {NoStop}%
\bibitem [{\citenamefont {Yamamoto}\ \emph {et~al.}(2001)\citenamefont {Yamamoto}, \citenamefont {Koashi},\ and\ \citenamefont {Imoto}}]{PhysRevA.64.012304}%
  \BibitemOpen
  \bibfield  {author} {\bibinfo {author} {\bibfnamefont {T.}~\bibnamefont {Yamamoto}}, \bibinfo {author} {\bibfnamefont {M.}~\bibnamefont {Koashi}},\ and\ \bibinfo {author} {\bibfnamefont {N.}~\bibnamefont {Imoto}},\ }\bibfield  {title} {\bibinfo {title} {Concentration and purification scheme for two partially entangled photon pairs},\ }\href {https://doi.org/10.1103/PhysRevA.64.012304} {\bibfield  {journal} {\bibinfo  {journal} {Phys. Rev. A}\ }\textbf {\bibinfo {volume} {64}},\ \bibinfo {pages} {012304} (\bibinfo {year} {2001})}\BibitemShut {NoStop}%
\bibitem [{\citenamefont {\ifmmode~\dot{Z}\else \.{Z}\fi{}ukowski}\ \emph {et~al.}(1993)\citenamefont {\ifmmode~\dot{Z}\else \.{Z}\fi{}ukowski}, \citenamefont {Zeilinger}, \citenamefont {Horne},\ and\ \citenamefont {Ekert}}]{PhysRevLett.71.4287}%
  \BibitemOpen
  \bibfield  {author} {\bibinfo {author} {\bibfnamefont {M.}~\bibnamefont {\ifmmode~\dot{Z}\else \.{Z}\fi{}ukowski}}, \bibinfo {author} {\bibfnamefont {A.}~\bibnamefont {Zeilinger}}, \bibinfo {author} {\bibfnamefont {M.~A.}\ \bibnamefont {Horne}},\ and\ \bibinfo {author} {\bibfnamefont {A.~K.}\ \bibnamefont {Ekert}},\ }\bibfield  {title} {\bibinfo {title} {``{Event}-ready-detectors'' {Bell} experiment via entanglement swapping},\ }\href {https://doi.org/10.1103/PhysRevLett.71.4287} {\bibfield  {journal} {\bibinfo  {journal} {Phys. Rev. Lett.}\ }\textbf {\bibinfo {volume} {71}},\ \bibinfo {pages} {4287} (\bibinfo {year} {1993})}\BibitemShut {NoStop}%
\bibitem [{\citenamefont {Altepeter}\ \emph {et~al.}(2005)\citenamefont {Altepeter}, \citenamefont {Jeffrey},\ and\ \citenamefont {Kwiat}}]{ALTEPETER2005105}%
  \BibitemOpen
  \bibfield  {author} {\bibinfo {author} {\bibfnamefont {J.}~\bibnamefont {Altepeter}}, \bibinfo {author} {\bibfnamefont {E.}~\bibnamefont {Jeffrey}},\ and\ \bibinfo {author} {\bibfnamefont {P.}~\bibnamefont {Kwiat}},\ }\bibfield  {title} {\bibinfo {title} {Photonic state tomography},\ }\href {https://doi.org/https://doi.org/10.1016/S1049-250X(05)52003-2} {\bibfield  {journal} {\bibinfo  {journal} {Advances In Atomic, Molecular, and Optical Physics}\ }\textbf {\bibinfo {volume} {52}},\ \bibinfo {pages} {105} (\bibinfo {year} {2005})}\BibitemShut {NoStop}%
\bibitem [{\citenamefont {Hradil}(1997)}]{PhysRevA.55.R1561}%
  \BibitemOpen
  \bibfield  {author} {\bibinfo {author} {\bibfnamefont {Z.}~\bibnamefont {Hradil}},\ }\bibfield  {title} {\bibinfo {title} {Quantum-state estimation},\ }\href {https://doi.org/10.1103/PhysRevA.55.R1561} {\bibfield  {journal} {\bibinfo  {journal} {Phys. Rev. A}\ }\textbf {\bibinfo {volume} {55}},\ \bibinfo {pages} {R1561} (\bibinfo {year} {1997})}\BibitemShut {NoStop}%
\bibitem [{\citenamefont {James}\ \emph {et~al.}(2001)\citenamefont {James}, \citenamefont {Kwiat}, \citenamefont {Munro},\ and\ \citenamefont {White}}]{PhysRevA.64.052312}%
  \BibitemOpen
  \bibfield  {author} {\bibinfo {author} {\bibfnamefont {D.~F.~V.}\ \bibnamefont {James}}, \bibinfo {author} {\bibfnamefont {P.~G.}\ \bibnamefont {Kwiat}}, \bibinfo {author} {\bibfnamefont {W.~J.}\ \bibnamefont {Munro}},\ and\ \bibinfo {author} {\bibfnamefont {A.~G.}\ \bibnamefont {White}},\ }\bibfield  {title} {\bibinfo {title} {Measurement of qubits},\ }\href {https://doi.org/10.1103/PhysRevA.64.052312} {\bibfield  {journal} {\bibinfo  {journal} {Phys. Rev. A}\ }\textbf {\bibinfo {volume} {64}},\ \bibinfo {pages} {052312} (\bibinfo {year} {2001})}\BibitemShut {NoStop}%
\bibitem [{\citenamefont {Qi}\ \emph {et~al.}(2017)\citenamefont {Qi}, \citenamefont {Hou}, \citenamefont {Wang}, \citenamefont {Dong}, \citenamefont {Zhong}, \citenamefont {Li}, \citenamefont {Xiang}, \citenamefont {Wiseman}, \citenamefont {Li},\ and\ \citenamefont {Guo}}]{Qi2017}%
  \BibitemOpen
  \bibfield  {author} {\bibinfo {author} {\bibfnamefont {B.}~\bibnamefont {Qi}}, \bibinfo {author} {\bibfnamefont {Z.}~\bibnamefont {Hou}}, \bibinfo {author} {\bibfnamefont {Y.}~\bibnamefont {Wang}}, \bibinfo {author} {\bibfnamefont {D.}~\bibnamefont {Dong}}, \bibinfo {author} {\bibfnamefont {H.-S.}\ \bibnamefont {Zhong}}, \bibinfo {author} {\bibfnamefont {L.}~\bibnamefont {Li}}, \bibinfo {author} {\bibfnamefont {G.-Y.}\ \bibnamefont {Xiang}}, \bibinfo {author} {\bibfnamefont {H.~M.}\ \bibnamefont {Wiseman}}, \bibinfo {author} {\bibfnamefont {C.-F.}\ \bibnamefont {Li}},\ and\ \bibinfo {author} {\bibfnamefont {G.-C.}\ \bibnamefont {Guo}},\ }\bibfield  {title} {\bibinfo {title} {Adaptive quantum state tomography via linear regression estimation: Theory and two-qubit experiment},\ }\href {https://doi.org/10.1038/s41534-017-0016-4} {\bibfield  {journal} {\bibinfo  {journal} {npj Quantum Information}\ }\textbf {\bibinfo {volume} {3}},\ \bibinfo {pages} {19} (\bibinfo {year} {2017})}\BibitemShut {NoStop}%
\bibitem [{\citenamefont {Guţă}\ \emph {et~al.}(2020)\citenamefont {Guţă}, \citenamefont {Kahn}, \citenamefont {Kueng},\ and\ \citenamefont {Tropp}}]{Guţă_2020}%
  \BibitemOpen
  \bibfield  {author} {\bibinfo {author} {\bibfnamefont {M.}~\bibnamefont {Guţă}}, \bibinfo {author} {\bibfnamefont {J.}~\bibnamefont {Kahn}}, \bibinfo {author} {\bibfnamefont {R.}~\bibnamefont {Kueng}},\ and\ \bibinfo {author} {\bibfnamefont {J.~A.}\ \bibnamefont {Tropp}},\ }\bibfield  {title} {\bibinfo {title} {Fast state tomography with optimal error bounds},\ }\href {https://doi.org/10.1088/1751-8121/ab8111} {\bibfield  {journal} {\bibinfo  {journal} {Journal of Physics A: Mathematical and Theoretical}\ }\textbf {\bibinfo {volume} {53}},\ \bibinfo {pages} {204001} (\bibinfo {year} {2020})}\BibitemShut {NoStop}%
\bibitem [{\citenamefont {Oka}\ \emph {et~al.}(2016)\citenamefont {Oka}, \citenamefont {Satoh},\ and\ \citenamefont {Van~Meter}}]{Oka2016}%
  \BibitemOpen
  \bibfield  {author} {\bibinfo {author} {\bibfnamefont {T.}~\bibnamefont {Oka}}, \bibinfo {author} {\bibfnamefont {T.}~\bibnamefont {Satoh}},\ and\ \bibinfo {author} {\bibfnamefont {R.}~\bibnamefont {Van~Meter}},\ }\bibfield  {title} {\bibinfo {title} {A {Classical} {Network} {Protocol} to {Support} {Distributed} {Quantum} {State} {Tomography}},\ }in\ \href {https://doi.org/10.1109/GLOCOMW.2016.7848802} {\emph {\bibinfo {booktitle} {2016 IEEE Globecom Workshops (GC Wkshps)}}}\ (\bibinfo {year} {2016})\ pp.\ \bibinfo {pages} {1--6}\BibitemShut {NoStop}%
\bibitem [{\citenamefont {Oka}(2018)}]{OkaBThesis}%
  \BibitemOpen
  \bibfield  {author} {\bibinfo {author} {\bibfnamefont {T.}~\bibnamefont {Oka}},\ }\emph {\bibinfo {title} {Optimized Fidelity Estimation in Purification for the Fastest Bootstrap of a Quantum Link}},\ \href {https://aqua.sfc.wide.ad.jp/publications/takafumi\_bachelors\_thesis.pdf} {\bibinfo {type} {Bachelor's thesis}},\ \bibinfo  {school} {Keio University} (\bibinfo {year} {2018})\BibitemShut {NoStop}%
\bibitem [{\citenamefont {Sugiyama}\ \emph {et~al.}(2013)\citenamefont {Sugiyama}, \citenamefont {Turner},\ and\ \citenamefont {Murao}}]{PhysRevLett.111.160406}%
  \BibitemOpen
  \bibfield  {author} {\bibinfo {author} {\bibfnamefont {T.}~\bibnamefont {Sugiyama}}, \bibinfo {author} {\bibfnamefont {P.~S.}\ \bibnamefont {Turner}},\ and\ \bibinfo {author} {\bibfnamefont {M.}~\bibnamefont {Murao}},\ }\bibfield  {title} {\bibinfo {title} {{Precision-Guaranteed Quantum Tomography}},\ }\href {https://doi.org/10.1103/PhysRevLett.111.160406} {\bibfield  {journal} {\bibinfo  {journal} {Phys. Rev. Lett.}\ }\textbf {\bibinfo {volume} {111}},\ \bibinfo {pages} {160406} (\bibinfo {year} {2013})}\BibitemShut {NoStop}%
\bibitem [{\citenamefont {Flammia}\ and\ \citenamefont {Liu}(2011)}]{PhysRevLett.106.230501}%
  \BibitemOpen
  \bibfield  {author} {\bibinfo {author} {\bibfnamefont {S.~T.}\ \bibnamefont {Flammia}}\ and\ \bibinfo {author} {\bibfnamefont {Y.-K.}\ \bibnamefont {Liu}},\ }\bibfield  {title} {\bibinfo {title} {{Direct Fidelity Estimation from Few Pauli Measurements}},\ }\href {https://doi.org/10.1103/PhysRevLett.106.230501} {\bibfield  {journal} {\bibinfo  {journal} {Phys. Rev. Lett.}\ }\textbf {\bibinfo {volume} {106}},\ \bibinfo {pages} {230501} (\bibinfo {year} {2011})}\BibitemShut {NoStop}%
\bibitem [{\citenamefont {da~Silva}\ \emph {et~al.}(2011)\citenamefont {da~Silva}, \citenamefont {Landon-Cardinal},\ and\ \citenamefont {Poulin}}]{PhysRevLett.107.210404}%
  \BibitemOpen
  \bibfield  {author} {\bibinfo {author} {\bibfnamefont {M.~P.}\ \bibnamefont {da~Silva}}, \bibinfo {author} {\bibfnamefont {O.}~\bibnamefont {Landon-Cardinal}},\ and\ \bibinfo {author} {\bibfnamefont {D.}~\bibnamefont {Poulin}},\ }\bibfield  {title} {\bibinfo {title} {Practical characterization of quantum devices without tomography},\ }\href {https://doi.org/10.1103/PhysRevLett.107.210404} {\bibfield  {journal} {\bibinfo  {journal} {Phys. Rev. Lett.}\ }\textbf {\bibinfo {volume} {107}},\ \bibinfo {pages} {210404} (\bibinfo {year} {2011})}\BibitemShut {NoStop}%
\bibitem [{\citenamefont {Eisert}\ \emph {et~al.}(2020)\citenamefont {Eisert}, \citenamefont {Hangleiter}, \citenamefont {Walk}, \citenamefont {Roth}, \citenamefont {Markham}, \citenamefont {Parekh}, \citenamefont {Chabaud},\ and\ \citenamefont {Kashefi}}]{Eisert:2020aa}%
  \BibitemOpen
  \bibfield  {author} {\bibinfo {author} {\bibfnamefont {J.}~\bibnamefont {Eisert}}, \bibinfo {author} {\bibfnamefont {D.}~\bibnamefont {Hangleiter}}, \bibinfo {author} {\bibfnamefont {N.}~\bibnamefont {Walk}}, \bibinfo {author} {\bibfnamefont {I.}~\bibnamefont {Roth}}, \bibinfo {author} {\bibfnamefont {D.}~\bibnamefont {Markham}}, \bibinfo {author} {\bibfnamefont {R.}~\bibnamefont {Parekh}}, \bibinfo {author} {\bibfnamefont {U.}~\bibnamefont {Chabaud}},\ and\ \bibinfo {author} {\bibfnamefont {E.}~\bibnamefont {Kashefi}},\ }\bibfield  {title} {\bibinfo {title} {Quantum certification and benchmarking},\ }\href {https://doi.org/10.1038/s42254-020-0186-4} {\bibfield  {journal} {\bibinfo  {journal} {Nature Reviews Physics}\ }\textbf {\bibinfo {volume} {2}},\ \bibinfo {pages} {382} (\bibinfo {year} {2020})}\BibitemShut {NoStop}%
\bibitem [{\citenamefont {Casapao}\ \emph {et~al.}(2025{\natexlab{a}})\citenamefont {Casapao}, \citenamefont {Maity}, \citenamefont {Benchasattabuse}, \citenamefont {Hajdu{\v s}ek}, \citenamefont {Van~Meter},\ and\ \citenamefont {Elkouss}}]{Casapao:2025aa}%
  \BibitemOpen
  \bibfield  {author} {\bibinfo {author} {\bibfnamefont {J.~C.~A.}\ \bibnamefont {Casapao}}, \bibinfo {author} {\bibfnamefont {A.~G.}\ \bibnamefont {Maity}}, \bibinfo {author} {\bibfnamefont {N.}~\bibnamefont {Benchasattabuse}}, \bibinfo {author} {\bibfnamefont {M.}~\bibnamefont {Hajdu{\v s}ek}}, \bibinfo {author} {\bibfnamefont {R.}~\bibnamefont {Van~Meter}},\ and\ \bibinfo {author} {\bibfnamefont {D.}~\bibnamefont {Elkouss}},\ }\bibfield  {title} {\bibinfo {title} {Disti-mator, an entanglement distillation-based state estimator},\ }\href {https://doi.org/10.1038/s42005-025-02352-2} {\bibfield  {journal} {\bibinfo  {journal} {Communications Physics}\ }\textbf {\bibinfo {volume} {8}},\ \bibinfo {pages} {461} (\bibinfo {year} {2025}{\natexlab{a}})}\BibitemShut {NoStop}%
\bibitem [{\citenamefont {Maity}\ \emph {et~al.}(2023)\citenamefont {Maity}, \citenamefont {Casapao}, \citenamefont {Benchasattabuse}, \citenamefont {Hajdusek}, \citenamefont {Van~Meter},\ and\ \citenamefont {Elkouss}}]{10.1145/3626570.3626594}%
  \BibitemOpen
  \bibfield  {author} {\bibinfo {author} {\bibfnamefont {A.~G.}\ \bibnamefont {Maity}}, \bibinfo {author} {\bibfnamefont {J.~C.~A.}\ \bibnamefont {Casapao}}, \bibinfo {author} {\bibfnamefont {N.}~\bibnamefont {Benchasattabuse}}, \bibinfo {author} {\bibfnamefont {M.}~\bibnamefont {Hajdusek}}, \bibinfo {author} {\bibfnamefont {R.}~\bibnamefont {Van~Meter}},\ and\ \bibinfo {author} {\bibfnamefont {D.}~\bibnamefont {Elkouss}},\ }\bibfield  {title} {\bibinfo {title} {Noise estimation in an entanglement distillation protocol},\ }\href {https://doi.org/10.1145/3626570.3626594} {\bibfield  {journal} {\bibinfo  {journal} {SIGMETRICS Perform. Eval. Rev.}\ }\textbf {\bibinfo {volume} {51}},\ \bibinfo {pages} {66} (\bibinfo {year} {2023})}\BibitemShut {NoStop}%
\bibitem [{\citenamefont {Casapao}\ \emph {et~al.}(2025{\natexlab{b}})\citenamefont {Casapao}, \citenamefont {Maity}, \citenamefont {Benchasattabuse}, \citenamefont {Hajdušek}, \citenamefont {Soeda}, \citenamefont {Van~Meter},\ and\ \citenamefont {Elkouss}}]{11250296}%
  \BibitemOpen
  \bibfield  {author} {\bibinfo {author} {\bibfnamefont {J.~C.~A.}\ \bibnamefont {Casapao}}, \bibinfo {author} {\bibfnamefont {A.~G.}\ \bibnamefont {Maity}}, \bibinfo {author} {\bibfnamefont {N.}~\bibnamefont {Benchasattabuse}}, \bibinfo {author} {\bibfnamefont {M.}~\bibnamefont {Hajdušek}}, \bibinfo {author} {\bibfnamefont {A.}~\bibnamefont {Soeda}}, \bibinfo {author} {\bibfnamefont {R.}~\bibnamefont {Van~Meter}},\ and\ \bibinfo {author} {\bibfnamefont {D.}~\bibnamefont {Elkouss}},\ }\bibfield  {title} {\bibinfo {title} {A double selection entanglement distillation-based state estimator},\ }in\ \href {https://doi.org/10.1109/QCE65121.2025.00122} {\emph {\bibinfo {booktitle} {2025 IEEE International Conference on Quantum Computing and Engineering (QCE)}}},\ Vol.~\bibinfo {volume} {01}\ (\bibinfo {year} {2025})\ pp.\ \bibinfo {pages} {1089--1097}\BibitemShut {NoStop}%
\bibitem [{\citenamefont {Van~Meter}\ \emph {et~al.}(2022)\citenamefont {Van~Meter}, \citenamefont {Satoh}, \citenamefont {Benchasattabuse}, \citenamefont {Teramoto}, \citenamefont {Matsuo}, \citenamefont {Hajdu{\v s}ek}, \citenamefont {Satoh}, \citenamefont {Nagayama},\ and\ \citenamefont {Suzuki}}]{9951258}%
  \BibitemOpen
  \bibfield  {author} {\bibinfo {author} {\bibfnamefont {R.}~\bibnamefont {Van~Meter}}, \bibinfo {author} {\bibfnamefont {R.}~\bibnamefont {Satoh}}, \bibinfo {author} {\bibfnamefont {N.}~\bibnamefont {Benchasattabuse}}, \bibinfo {author} {\bibfnamefont {K.}~\bibnamefont {Teramoto}}, \bibinfo {author} {\bibfnamefont {T.}~\bibnamefont {Matsuo}}, \bibinfo {author} {\bibfnamefont {M.}~\bibnamefont {Hajdu{\v s}ek}}, \bibinfo {author} {\bibfnamefont {T.}~\bibnamefont {Satoh}}, \bibinfo {author} {\bibfnamefont {S.}~\bibnamefont {Nagayama}},\ and\ \bibinfo {author} {\bibfnamefont {S.}~\bibnamefont {Suzuki}},\ }\bibfield  {title} {\bibinfo {title} {A quantum internet architecture},\ }in\ \href {https://doi.org/10.1109/QCE53715.2022.00055} {\emph {\bibinfo {booktitle} {2022 IEEE International Conference on Quantum Computing and Engineering (QCE)}}}\ (\bibinfo {year} {2022})\ pp.\ \bibinfo {pages} {341--352}\BibitemShut {NoStop}%
\bibitem [{\citenamefont {Kozlowski}\ \emph {et~al.}(2023)\citenamefont {Kozlowski}, \citenamefont {Wehner}, \citenamefont {Van{ }Meter}, \citenamefont {Rijsman}, \citenamefont {Cacciapuoti}, \citenamefont {Caleffi},\ and\ \citenamefont {Nagayama}}]{rfc9340}%
  \BibitemOpen
  \bibfield  {author} {\bibinfo {author} {\bibfnamefont {W.}~\bibnamefont {Kozlowski}}, \bibinfo {author} {\bibfnamefont {S.}~\bibnamefont {Wehner}}, \bibinfo {author} {\bibfnamefont {R.}~\bibnamefont {Van{ }Meter}}, \bibinfo {author} {\bibfnamefont {B.}~\bibnamefont {Rijsman}}, \bibinfo {author} {\bibfnamefont {A.~S.}\ \bibnamefont {Cacciapuoti}}, \bibinfo {author} {\bibfnamefont {M.}~\bibnamefont {Caleffi}},\ and\ \bibinfo {author} {\bibfnamefont {S.}~\bibnamefont {Nagayama}},\ }\href {https://doi.org/10.17487/RFC9340} {\bibinfo {title} {{Architectural Principles for a Quantum Internet}}},\ \bibinfo {howpublished} {RFC 9340} (\bibinfo {year} {2023})\BibitemShut {NoStop}%
\bibitem [{\citenamefont {Wehner}\ \emph {et~al.}(2018)\citenamefont {Wehner}, \citenamefont {Elkouss},\ and\ \citenamefont {Hanson}}]{Wehner18:eaam9288}%
  \BibitemOpen
  \bibfield  {author} {\bibinfo {author} {\bibfnamefont {S.}~\bibnamefont {Wehner}}, \bibinfo {author} {\bibfnamefont {D.}~\bibnamefont {Elkouss}},\ and\ \bibinfo {author} {\bibfnamefont {R.}~\bibnamefont {Hanson}},\ }\bibfield  {title} {\bibinfo {title} {Quantum internet: A vision for the road ahead},\ }\bibfield  {journal} {\bibinfo  {journal} {Science}\ }\textbf {\bibinfo {volume} {362}},\ \href {https://doi.org/10.1126/science.aam9288} {10.1126/science.aam9288} (\bibinfo {year} {2018})\BibitemShut {NoStop}%
\bibitem [{\citenamefont {Krutyanskiy}\ \emph {et~al.}(2023)\citenamefont {Krutyanskiy}, \citenamefont {Canteri}, \citenamefont {Meraner}, \citenamefont {Bate}, \citenamefont {Krcmarsky}, \citenamefont {Schupp}, \citenamefont {Sangouard},\ and\ \citenamefont {Lanyon}}]{PhysRevLett.130.213601}%
  \BibitemOpen
  \bibfield  {author} {\bibinfo {author} {\bibfnamefont {V.}~\bibnamefont {Krutyanskiy}}, \bibinfo {author} {\bibfnamefont {M.}~\bibnamefont {Canteri}}, \bibinfo {author} {\bibfnamefont {M.}~\bibnamefont {Meraner}}, \bibinfo {author} {\bibfnamefont {J.}~\bibnamefont {Bate}}, \bibinfo {author} {\bibfnamefont {V.}~\bibnamefont {Krcmarsky}}, \bibinfo {author} {\bibfnamefont {J.}~\bibnamefont {Schupp}}, \bibinfo {author} {\bibfnamefont {N.}~\bibnamefont {Sangouard}},\ and\ \bibinfo {author} {\bibfnamefont {B.~P.}\ \bibnamefont {Lanyon}},\ }\bibfield  {title} {\bibinfo {title} {{Telecom-Wavelength Quantum Repeater Node Based on a Trapped-Ion Processor}},\ }\href {https://doi.org/10.1103/PhysRevLett.130.213601} {\bibfield  {journal} {\bibinfo  {journal} {Phys. Rev. Lett.}\ }\textbf {\bibinfo {volume} {130}},\ \bibinfo {pages} {213601} (\bibinfo {year} {2023})}\BibitemShut {NoStop}%
\bibitem [{\citenamefont {Pompili}\ \emph {et~al.}(2021)\citenamefont {Pompili}, \citenamefont {Hermans}, \citenamefont {Baier}, \citenamefont {Beukers}, \citenamefont {Humphreys}, \citenamefont {Schouten}, \citenamefont {Vermeulen}, \citenamefont {Tiggelman}, \citenamefont {dos Santos~Martins}, \citenamefont {Dirkse}, \citenamefont {Wehner},\ and\ \citenamefont {Hanson}}]{doi:10.1126/science.abg1919}%
  \BibitemOpen
  \bibfield  {author} {\bibinfo {author} {\bibfnamefont {M.}~\bibnamefont {Pompili}}, \bibinfo {author} {\bibfnamefont {S.~L.~N.}\ \bibnamefont {Hermans}}, \bibinfo {author} {\bibfnamefont {S.}~\bibnamefont {Baier}}, \bibinfo {author} {\bibfnamefont {H.~K.~C.}\ \bibnamefont {Beukers}}, \bibinfo {author} {\bibfnamefont {P.~C.}\ \bibnamefont {Humphreys}}, \bibinfo {author} {\bibfnamefont {R.~N.}\ \bibnamefont {Schouten}}, \bibinfo {author} {\bibfnamefont {R.~F.~L.}\ \bibnamefont {Vermeulen}}, \bibinfo {author} {\bibfnamefont {M.~J.}\ \bibnamefont {Tiggelman}}, \bibinfo {author} {\bibfnamefont {L.}~\bibnamefont {dos Santos~Martins}}, \bibinfo {author} {\bibfnamefont {B.}~\bibnamefont {Dirkse}}, \bibinfo {author} {\bibfnamefont {S.}~\bibnamefont {Wehner}},\ and\ \bibinfo {author} {\bibfnamefont {R.}~\bibnamefont {Hanson}},\ }\bibfield  {title} {\bibinfo {title} {Realization of a multinode quantum network of remote solid-state qubits},\ }\href {https://doi.org/10.1126/science.abg1919} {\bibfield  {journal}
  {\bibinfo  {journal} {Science}\ }\textbf {\bibinfo {volume} {372}},\ \bibinfo {pages} {259} (\bibinfo {year} {2021})}\BibitemShut {NoStop}%
\bibitem [{\citenamefont {Yin}\ \emph {et~al.}(2017)\citenamefont {Yin}, \citenamefont {Cao}, \citenamefont {Li}, \citenamefont {Liao}, \citenamefont {Zhang}, \citenamefont {Ren}, \citenamefont {Cai}, \citenamefont {Liu}, \citenamefont {Li}, \citenamefont {Dai}, \citenamefont {Li}, \citenamefont {Lu}, \citenamefont {Gong}, \citenamefont {Xu}, \citenamefont {Li}, \citenamefont {Li}, \citenamefont {Yin}, \citenamefont {Jiang}, \citenamefont {Li}, \citenamefont {Jia}, \citenamefont {Ren}, \citenamefont {He}, \citenamefont {Zhou}, \citenamefont {Zhang}, \citenamefont {Wang}, \citenamefont {Chang}, \citenamefont {Zhu}, \citenamefont {Liu}, \citenamefont {Chen}, \citenamefont {Lu}, \citenamefont {Shu}, \citenamefont {Peng}, \citenamefont {Wang},\ and\ \citenamefont {Pan}}]{doi:10.1126/science.aan3211}%
  \BibitemOpen
  \bibfield  {author} {\bibinfo {author} {\bibfnamefont {J.}~\bibnamefont {Yin}}, \bibinfo {author} {\bibfnamefont {Y.}~\bibnamefont {Cao}}, \bibinfo {author} {\bibfnamefont {Y.-H.}\ \bibnamefont {Li}}, \bibinfo {author} {\bibfnamefont {S.-K.}\ \bibnamefont {Liao}}, \bibinfo {author} {\bibfnamefont {L.}~\bibnamefont {Zhang}}, \bibinfo {author} {\bibfnamefont {J.-G.}\ \bibnamefont {Ren}}, \bibinfo {author} {\bibfnamefont {W.-Q.}\ \bibnamefont {Cai}}, \bibinfo {author} {\bibfnamefont {W.-Y.}\ \bibnamefont {Liu}}, \bibinfo {author} {\bibfnamefont {B.}~\bibnamefont {Li}}, \bibinfo {author} {\bibfnamefont {H.}~\bibnamefont {Dai}}, \bibinfo {author} {\bibfnamefont {G.-B.}\ \bibnamefont {Li}}, \bibinfo {author} {\bibfnamefont {Q.-M.}\ \bibnamefont {Lu}}, \bibinfo {author} {\bibfnamefont {Y.-H.}\ \bibnamefont {Gong}}, \bibinfo {author} {\bibfnamefont {Y.}~\bibnamefont {Xu}}, \bibinfo {author} {\bibfnamefont {S.-L.}\ \bibnamefont {Li}}, \bibinfo {author} {\bibfnamefont {F.-Z.}\ \bibnamefont {Li}}, \bibinfo {author}
  {\bibfnamefont {Y.-Y.}\ \bibnamefont {Yin}}, \bibinfo {author} {\bibfnamefont {Z.-Q.}\ \bibnamefont {Jiang}}, \bibinfo {author} {\bibfnamefont {M.}~\bibnamefont {Li}}, \bibinfo {author} {\bibfnamefont {J.-J.}\ \bibnamefont {Jia}}, \bibinfo {author} {\bibfnamefont {G.}~\bibnamefont {Ren}}, \bibinfo {author} {\bibfnamefont {D.}~\bibnamefont {He}}, \bibinfo {author} {\bibfnamefont {Y.-L.}\ \bibnamefont {Zhou}}, \bibinfo {author} {\bibfnamefont {X.-X.}\ \bibnamefont {Zhang}}, \bibinfo {author} {\bibfnamefont {N.}~\bibnamefont {Wang}}, \bibinfo {author} {\bibfnamefont {X.}~\bibnamefont {Chang}}, \bibinfo {author} {\bibfnamefont {Z.-C.}\ \bibnamefont {Zhu}}, \bibinfo {author} {\bibfnamefont {N.-L.}\ \bibnamefont {Liu}}, \bibinfo {author} {\bibfnamefont {Y.-A.}\ \bibnamefont {Chen}}, \bibinfo {author} {\bibfnamefont {C.-Y.}\ \bibnamefont {Lu}}, \bibinfo {author} {\bibfnamefont {R.}~\bibnamefont {Shu}}, \bibinfo {author} {\bibfnamefont {C.-Z.}\ \bibnamefont {Peng}}, \bibinfo {author} {\bibfnamefont {J.-Y.}\
  \bibnamefont {Wang}},\ and\ \bibinfo {author} {\bibfnamefont {J.-W.}\ \bibnamefont {Pan}},\ }\bibfield  {title} {\bibinfo {title} {Satellite-based entanglement distribution over 1200 kilometers},\ }\href {https://doi.org/10.1126/science.aan3211} {\bibfield  {journal} {\bibinfo  {journal} {Science}\ }\textbf {\bibinfo {volume} {356}},\ \bibinfo {pages} {1140} (\bibinfo {year} {2017})}\BibitemShut {NoStop}%
\bibitem [{\citenamefont {Pfaff}\ \emph {et~al.}(2014)\citenamefont {Pfaff}, \citenamefont {Hensen}, \citenamefont {Bernien}, \citenamefont {van Dam}, \citenamefont {Blok}, \citenamefont {Taminiau}, \citenamefont {Tiggelman}, \citenamefont {Schouten}, \citenamefont {Markham}, \citenamefont {Twitchen},\ and\ \citenamefont {Hanson}}]{Pfaff29052014}%
  \BibitemOpen
  \bibfield  {author} {\bibinfo {author} {\bibfnamefont {W.}~\bibnamefont {Pfaff}}, \bibinfo {author} {\bibfnamefont {B.}~\bibnamefont {Hensen}}, \bibinfo {author} {\bibfnamefont {H.}~\bibnamefont {Bernien}}, \bibinfo {author} {\bibfnamefont {S.~B.}\ \bibnamefont {van Dam}}, \bibinfo {author} {\bibfnamefont {M.~S.}\ \bibnamefont {Blok}}, \bibinfo {author} {\bibfnamefont {T.~H.}\ \bibnamefont {Taminiau}}, \bibinfo {author} {\bibfnamefont {M.~J.}\ \bibnamefont {Tiggelman}}, \bibinfo {author} {\bibfnamefont {R.~N.}\ \bibnamefont {Schouten}}, \bibinfo {author} {\bibfnamefont {M.}~\bibnamefont {Markham}}, \bibinfo {author} {\bibfnamefont {D.~J.}\ \bibnamefont {Twitchen}},\ and\ \bibinfo {author} {\bibfnamefont {R.}~\bibnamefont {Hanson}},\ }\bibfield  {title} {\bibinfo {title} {Unconditional quantum teleportation between distant solid-state quantum bits},\ }\bibfield  {journal} {\bibinfo  {journal} {Science}\ }\href {https://doi.org/10.1126/science.1253512} {10.1126/science.1253512} (\bibinfo {year}
  {2014})\BibitemShut {NoStop}%
\bibitem [{\citenamefont {Cuomo}\ \emph {et~al.}(2020)\citenamefont {Cuomo}, \citenamefont {Caleffi},\ and\ \citenamefont {Cacciapuoti}}]{Cuomo:2020aa}%
  \BibitemOpen
  \bibfield  {author} {\bibinfo {author} {\bibfnamefont {D.}~\bibnamefont {Cuomo}}, \bibinfo {author} {\bibfnamefont {M.}~\bibnamefont {Caleffi}},\ and\ \bibinfo {author} {\bibfnamefont {A.~S.}\ \bibnamefont {Cacciapuoti}},\ }\bibfield  {title} {\bibinfo {title} {Towards a distributed quantum computing ecosystem},\ }\href {https://doi.org/https://doi.org/10.1049/iet-qtc.2020.0002} {\bibfield  {journal} {\bibinfo  {journal} {IET Quantum Communication}\ }\textbf {\bibinfo {volume} {1}},\ \bibinfo {pages} {3} (\bibinfo {year} {2020})}\BibitemShut {NoStop}%
\bibitem [{\citenamefont {Huang}\ \emph {et~al.}(2020)\citenamefont {Huang}, \citenamefont {Kueng},\ and\ \citenamefont {Preskill}}]{Huang:2020aa}%
  \BibitemOpen
  \bibfield  {author} {\bibinfo {author} {\bibfnamefont {H.-Y.}\ \bibnamefont {Huang}}, \bibinfo {author} {\bibfnamefont {R.}~\bibnamefont {Kueng}},\ and\ \bibinfo {author} {\bibfnamefont {J.}~\bibnamefont {Preskill}},\ }\bibfield  {title} {\bibinfo {title} {Predicting many properties of a quantum system from very few measurements},\ }\href {https://doi.org/10.1038/s41567-020-0932-7} {\bibfield  {journal} {\bibinfo  {journal} {Nature Physics}\ }\textbf {\bibinfo {volume} {16}},\ \bibinfo {pages} {1050} (\bibinfo {year} {2020})}\BibitemShut {NoStop}%
\end{thebibliography}%

\end{document}